\documentclass[leqno,11pt]{amsart}
\usepackage[utf8]{inputenc}
\usepackage{enumitem}
\usepackage[dvipsnames]{xcolor}

\usepackage[colorinlistoftodos,prependcaption,textsize=tiny]{todonotes}

\usepackage{amsmath,amssymb,mathrsfs,mathtools,verbatim, amsthm, amsfonts, amssymb}
\usepackage{mathtools}
\usepackage{xfrac}    
\usepackage{tikz}
\usepackage{tikz-cd}
\usepackage{graphicx}
\usepackage{multimedia}
\usepackage{float}
\usepackage[margin=1in]{geometry}
\usepackage[onehalfspacing]{setspace}
\usepgflibrary{shapes.geometric}

\theoremstyle{plain}
\newtheorem*{namedthm}{\namedthmname}
\newcounter{namedthm}

\newtheorem{result}{Result}
\newtheorem{conjecture}{Conjecture}
\newtheorem{lemma}{Lemma}
\newtheorem{theorem}{Theorem}
\newtheorem{proposition}{Proposition}
\newtheorem{remark}{Remark}



\newcommand{\betti}[1] {{\beta}_{#1}}

\newcommand{\C}{\mathbb{C}}
\newcommand{\R}{\mathbb{R}}
\newcommand{\Sp}{\mathbb{S}}

\newcommand{\Z}{\mathbb{Z}}
\newcommand{\N}{\mathbb{N}}

\newcommand{\ii}{\mathbf{i}}

\newcommand{\norm}[1] {\|{#1}\|}

\newcommand{\dist}[2] {\|{#1}-{#2}\|}
\newcommand{\diff}[1] {\,{\rm d}{#1}}

\newcommand{\ee}            {\varepsilon}

\newcommand{\Skip}[1]       {}

\title{Counting Equilibria of the Electrostatic Potential}

\author{Herbert Edelsbrunner}
\address{ISTA (Institute of Science and Technology Austria), Kloster\-neu\-burg, Austria}
\email{Herbert.Edelsbrunner@ista.ac.at}

\author{Christopher Fillmore}
\address{ISTA (Institute of Science and Technology Austria), Kloster\-neu\-burg, Austria}
\email{cdfillmore@gmail.com}

\author{Gonçalo Oliveira}
\address{IST (Instituto Superior T\'ecnico), Universidade de Lisboa, Portugal}
\email{galato97@gmail.com}


\begin{document}

	\begin{abstract}
		In 1873, James C. Maxwell \cite{Max54} conjectured that the electric field generated by $n$ point charges in generic position has at most $(n-1)^2$ isolated zeroes.
		
		The first (non-optimal) upper bound was only obtained in 2007 by Gabrielov, Novikov and Shapiro, who also posed two additional interesting conjectures \cite{GNS07}. 
		
		In this article we give the best upper bound known to date on the number of zeroes of the electric field, and construct a counterexample to Conjecture~1.8 in \cite{GNS07} that the number of equilibria cannot exceed those of the distance function defined by the unit point charges.
		
		Finally, we note that it is quite possible that Maxwell's quadratic upper bound is not tight, so it is prudent to find lower bounds. Hence, we also explore examples and construct configurations of charges achieving the highest ratios of the number of electric field zeroes by point charges found to this day.
	\end{abstract}

	\maketitle

	\section{Introduction}
	\label{sec:1}
	
	\noindent \textbf{Context.}
	The \emph{electrostatic potential} generated by $n$ point charges located at $A_1, A_2, \ldots, A_n \in \R^3$ of magnitudes $\zeta_1, \zeta_2, \ldots, \zeta_n \in \R$ is the function $V \colon \R^3 \setminus \{A_1,A_2,\ldots,A_n\} \to \R$ defined by
	\begin{align}
		V (x) &= \sum\nolimits_{i=1}^n \frac{\zeta_i}{\dist{x}{A_i}},
		\label{eqn:V}
	\end{align}
	in which $\dist{x}{A_i}$ is the Euclidean distance between the two points.
	The \emph{electric field} is the gradient of this potential.
	In this article, we shall focus in the case in which all the charges $\zeta_i$ are positive.
	The zeros of the electric field are the \emph{equilibria} or the \emph{critical points} of the potential.
	They are also sometimes referred to as \emph{electrostatic points} as they are the points at which the electric field vanishes, thus corresponding to equilibrium positions for test charges.
	
	\smallskip
	A well-known Morse theoretic argument reviewed later in Section~\ref{sec:2}, shows that, for the generic location of $n$ point charges, the number of equilibria is at least $n-1$.
	For example, this is achieved by placing the point charges along a line. 
	One may wonder which configurations of charges generate the maximum number of equilibria. 
	This is at the core of the 1873 conjecture posed by Maxwell in his famous ``A Treatise on Electricity and Magnetism''.
	\begin{conjecture}[Maxwell in \cite{Max54}]
		\label{conj:Maxwell}
		The electrostatic potential defined by $n$ point charges in $\R^3$ has at most $(n-1)^2$ equilibria.
	\end{conjecture}
	For $n=3$, this conjectured upper bound is attained by the vertices of an equilateral triangle in $\R^3$.
	Placing a unit charge at each of these three points, we get three $2$-saddles and one $1$-saddle, so a total of four equilibria of $V$. 
	
	\smallskip
	A general upper bound on the number of equilibria was only obtained in 2007 \cite{GNS07}, by using Kovanskii's theory of Fewnomials \cite{K}. However, it grows much faster than the quadratic upper bound conjectured by Maxwell. Indeed, it is given by $2^{2n^2}(3n)^{2n}$, which already for $n=3$ is $139314069504$. More recently, using the Thom-Milnor theorem, this was further improved in \cite{Z} to $5 \times 9^{3+n}$. As we shall see further below, one of our contributions in this article will be to improve this bound even further by a clever use of Bezout's theorem.
	
	Beyond the upper bound, in \cite{GNS07} Gabrielov, Novikov and Shapiro prove interesting results and pose very stimulating conjectures, which we will now explain and later address in this article.
	An important auxiliary tool in \cite{GNS07} is the family of potentials
	\begin{align}
		V_p(x) &= \sum_{i=1}^n \frac{\zeta_i^p}{\norm{x-A_i}^p},
		\hspace{0.1in} \mbox{\rm which satisfies} \hspace{0.1in}
		\lim_{p \to \infty} \frac{1}{\sqrt[p]{V_p (x)}}
		= \min_{1 \leq i \leq n} \frac{\norm{x-A_i}}{\zeta_i}.
		\label{eqn:Vp-intro}
	\end{align}
	The authors prove in Theorem 1.7 of that reference that for sufficiently large $p$, the number of equilibria of $V_p$ of any given index remains constant with increasing $p$.\footnote{In fact, said theorem proves that for sufficiently large $p$, the equilibria of $V_p$ are in correspondence with the so-called effective cells of positive codimension in the corresponding Voronoi diagram. Furthermore, the Morse index of such equilibria coincides with the dimension of the corresponding cell.} 
	Then, for $j \in \{0,1,2,3\}$, let $\#_j$ denote this limiting number of equilibria of $V_p$ for $p \gg 1$. 
	This result, plus the expectation that the number of equilibria of $V_p$ is non-decreasing with $p$, for all $p\geq 1$ leads to the following conjecture.
	\begin{conjecture}[Conjecture 1.8 (a) in \cite{GNS07}]
		\label{conj:1.8}
		For $n$ unit point charges in generic position, and all $p \geq 1$, the number of index $j$ equilibria of $V_p$ is at most $\#_j$.
	\end{conjecture}
	As we shall see later in this introduction, we will provide a counterexample to this conjecture. On the other hand, for even $p$, one can easily obtain a linear upper bound on the number of equilibria of the restriction of $V_p$ to any line in Euclidean space. A stronger version of this statement was also conjectured in \cite{GNS07} as we shall now state. 
	
	\begin{conjecture}[Conjecture 1.9 in \cite{GNS07}]
		\label{conj:1.9}
		For all $p \geq 1$, the number of equilibria of the restriction of $V_p$ to any straight line in $\R^3$ is at most $2n-1$.
	\end{conjecture}

	\medskip \noindent \textbf{Main Results.}
	As stated in the abstract, the key contributions of this article are:
	
	\begin{itemize}
		\item an improved upper bound on the number of isolated equilibria for the electrostatic potential;
		\item a new record for the ratio of the number of equilibria over the number of point charges;
		\item a counterexample to Conjecture~\ref{conj:1.8}; and
		\item a number of other interesting results and explorations.
	\end{itemize}
	
	We shall now start by stating our first main result, namely the new upper bound on the number of isolated equilibria for the electrostatic potential.
	
	\begin{result}\label{thm:Main}
		Let $n \in \mathbb{N}$, $A_2, \ldots , A_n \in \mathbb{R}^3$, and $\zeta_1, \ldots ,\zeta_n \in \mathbb{R} \backslash \{0\}$. Then, for the generic position of $A_1 \in \mathbb{R}^3$, the number of isolated critical points of $V$ is at most
		$$2^n \times (3n-2)^3 .$$
	\end{result}
	
	The proof of this result appears in Section \ref{sec:Upper bound}, together with an analogous but stronger result for the potentials $V_p$ for even $p$. For these, an easier application of the same method yields a polynomial upper bound.
	
	In Sections~\ref{sec:2}, we give a general discussion and exploration of examples obtained by placing unit charges at the vertices of several polyhedra. 
	It is here that we arrive at our first main result, which is a family of examples constructed in Sections \ref{sec:2.3} and \ref{sec:2.4}. 
	The conclusion of that construction is the following.
	\begin{result}[New record for number of equilibria]
		\label{result:Antiprism}  
		For any $\ee>0$, there are a sufficiently large integer, $n$, and point charges, $\{A_1, A_2, \ldots, A_n\} \subseteq \R^3$, such that
		\begin{align}
			V(x) &= \sum\nolimits_{i=1}^n \frac{1}{\norm{x-A_i}},
		\end{align}
		has $k$ equilibria with $\frac{k}{n} > \frac{25}{7}-\ee$.
	\end{result}
	Our construction gives a precise family of examples which, for each $\ell \in \N$ constructs a configuration of $n_\ell = 8^\ell$ points with unit charges in $\R^3$ that define an electrostatic potential with $\frac{25}{7} (n_\ell-1)$ equilibria. Recently, there has been interest in constructing closed geodesics on Calabi-Yau manifolds \cite{Douglas}, and this result may find applications to that purpose \cite{Oli}.
	
	\smallskip
	In Section~\ref{sec:3}, we study the potentials $V_p$ defined in \eqref{eqn:Vp-intro} and, after deducing some results, culminate with our second main result which is a counterexample to \cite[Conjecture~1.8(a)]{GNS07}.
	The conclusion of the discussion in Section~\ref{sec:3.2} can be summarised as follows.
	
	\begin{result}[Counterexample to Conjecture \ref{conj:1.8}]
		\label{result:Counterexample}
		Let $A_1$ to $A_{24}$ be the unit point charges at the vertices of the truncated octahedron---the Voronoi domain in the body centered cubic lattice---and consider the resulting $1$-parameter family of potentials
		\begin{align}
			V_p(x) &= \sum\nolimits_{i=1}^n \frac{1}{\norm{x-A_i}^p}.
			\label{eqn:result2}
		\end{align}
		Then, there is $p_*$ such the number of index-$1$ equilibria of $V=V_{1}$ is greater than that of $V_p$ for all $p>p_*$; that is: greater than $\#_1$.
	\end{result}
	In the above counterexample, the number of index-$2$ equilibria is equal to $\#_2$, so even the total number of equilibria of $V$ is larger than that of the limiting Euclidean distance function.
	
	\smallskip
	In addition to the main results, this article contains a number of other contributions, including a thorough exploration of the electrostatic potential generated by unit point charges positioned at the vertices of several interesting convex polyhedra.
	For example, we observe the following curious phenomenon: the configurations obtained by placing the unit charges at the vertices of Platonic and Catalan polyhedra never generate electric potentials with $1$-saddles, while for Archimedean polyhedra---which are dual to the Catalan ones---such $1$-saddles are always present. We also observe, for even $p$, the upper bound of $p(n-1)+ 2n-1$ on the number of equilibria of the restriction of $V_p$ defined by $n$ unit point charges to any line in $\R^3$.

	
	\section{The upper bound}\label{sec:Upper bound}
	
	This section contains a proof of the upper bound stated in our first main result \ref{thm:Main} is as follows. We start by computing the critical point equation and clearing out denominators in Lemma \ref{lem:Clearing out denominators}. The resulting equations can be turned into polynomial ones by a trick that formally introduces extra variables to encode the quantities given by square roots. This is the content of section \ref{sec:Polynomial}, where the key polynomial maps are defined. 
	
	To give an upper bound on the number of isolated solutions to these polynomial equations, we use Bezout's theorem, which gives such a bound from the degrees of the relevant polynomial equations. However, there is a subtlety involved in using this upper bound. It requires showing that non-degenerate critical points of $V$ correspond to non-isolated zeroes of the polynomial system. This is proven in \ref{ss:Isolated zeroes} and these preliminary arguments are then all put together in \ref{ss:proof} to finalize the proof of Theorem \ref{thm:Main}.

	\subsection{Towards polynomial equations}\label{sec:Polynomial}
	
	The electric potential generated by charges of magnitude $\zeta_i$ located at the points $p_i$, for $i=1, \ldots, n$, is
	$$V(x)=\sum_{i=1}^{n} \frac{\zeta_i}{r_i},$$
	where $r_i=|x-p_i|$. We shall write $p_i=(p_{i1},p_{i2},p_{i3})$ and $x=(x_1,x_2,x_3)$.
	
	\begin{lemma}\label{lem:Clearing out denominators}
		For $j=1,2,3$, let
		$$R_j(x)=\sum_{i=1}^{n} \zeta_i (x_j-p_{ij}) \prod_{\ell \neq i} r_\ell^3.$$
		A point $x\in \mathbb{R}^3$ is a critical point of $V$  if and only if it is a common zero of $R_1$, $R_2$, and $R_3$.
	\end{lemma}
	\begin{proof}
		We can compute that $\partial_{x_j} r_i = \frac{x_j-p_{ij}}{r_i}$, then $\partial_{x_j} r_i^{-1} = -  \frac{x_j-p_{ij}}{r_i^3}$ and therefore
		$$\partial_{x_j} V(x) = - \sum_{i=1}^{n} \frac{\zeta_i}{r_i^3} (x_j-p_{ij}) . $$
		We can clear out denominators as follows
		\begin{align*}
			\partial_{x_j} V(x) & = - \sum_{i=1}^{n} \zeta_i (x_j-p_{ij}) \times \frac{\prod_{\ell \neq i} r_\ell^3}{\prod_{m=1}^n r_m^3}  \\
			& = - \frac{ \sum_{i=1}^{n} \zeta_i (x_j-p_{ij}) \prod_{\ell \neq i} r_\ell^3}{\prod_{m=1}^n r_m^3}.
		\end{align*}
	\end{proof}
	
	At this point we introduce the variables $u_1 , \ldots , u_n$ and rewrite the system from Lemma \ref{lem:Clearing out denominators} as the zeros of the following polynomial functions
	\begin{align*}
		\mathcal{R}_j(x,u) & = \sum_{i=1}^{n} \zeta_i (x_j-p_{ij}) \prod_{\ell \neq i} u_\ell^3 , \ \ \text{for} \ j \in \{1,2,3\}, \\
		\mathcal{Q}_m(x,u) & = u_m^2 - |x-p_m|^2 , \ \ \text{for} \ m \in \{1, \ldots ,n\}.
	\end{align*}
	Hence, $x$ is a critical point of $V$ if and only if there is $u \in \mathbb{R}^n$ such that $(x,u)$ is a zero of the polynomial map
	$$\mathcal{P}: \mathbb{R}^{3} \times \mathbb{R}^n \to \mathbb{R}^{3} \times \mathbb{R}^n$$ 
	given by 
	\begin{equation}\label{eq:Polynomial Map}
		\mathcal{P}(x,u) = \left( \mathcal{R}(x,u) , \mathcal{Q}(x,u) \right) .
	\end{equation}
	As polynomials in the variables $(x,u) \in \mathbb{R}^3 \times \mathbb{R}^n$, its components have degrees 
	$$\deg(\mathcal{R}_j)=1+3(n-1)=3n-2,$$
	for $j \in \{1,2,3\}$, and
	$$\deg(\mathcal{Q}_m)=2,$$
	for $m \in \{1, \ldots , n\}$.

	\subsection{Bezout's theorem on affine hypersurfaces}\label{sec:Bezout}
	
	The main tool we shall use is the following theorem.
	
	\begin{theorem}[Affine version of Bezout's theorem]\label{thm:Bezout}
		Let $P_1, \ldots , P_m \in \mathbb{R}[x_1, \ldots , x_m]$ be polynomials of degree $d_1 , \ldots , d_d$ respectively. Then, the number of non-degenerate zeroes of the system $P_1=0, \ldots , P_m=0$ is at most $d_1 \ldots d_m$.
	\end{theorem}
	
	A proof of this result for complex polynomials is given in \cite{RW} which states the result for complex polynomials, though the upper bound can be obtained by extending the polynomials to become complex (a non-degenerate zero of a real polynomial equation is also non-degenerate as a zero of its complexification). Alternatively, this result, as stated, is a particular case of the Theorem stated in page 12 of \cite{K}. This is the result used in \cite{GNS07} to give the previously found bound. However, here we are able to substantially improve it by working with polynomial, rather quasi-polynomial, equations.

	\subsection{Isolated zeroes}\label{ss:Isolated zeroes}
	
	The affine Bezout's theorem will be used to give an upper bound on the number of isolated zeroes of the polynomial map $\mathcal{P}$ defined in \ref{eq:Polynomial Map}. However, for this to give an effective bound on the number of non-degenerate zeroes of the electric potential $V$ we must show that they correspond to isolated zeroes of this polynomial map $P$. Our next result shows this to be the case if one of the charges charges can be chosen in generic position..
	
	\begin{lemma}\label{lem:Isolated zeroes v0}
		For the generic choice of $p_1$, any critical point $x$ of $V$ is such that $(x,r)$ is an isolated zero of $\mathcal{P}$ for $r=(r_1, \ldots , r_n)$.
	\end{lemma}
	\begin{proof}
		The result will be proven by showing that, for generic $p_1$ and $x$ a critical point of $V$, the derivative of the map $\mathcal{P}$ has full rank, i.e $n+3$. In the standard basis of $\mathbb{R}^{3+n}$, such derivative is represented by the map
		$$\begin{pmatrix}
			\left(\frac{\partial \mathcal{R}_j}{\partial x_s}\right)_{js} & | & \left(\frac{\partial \mathcal{R}_j}{\partial u_q} \right)_{jq} \\
			---- & | & ---- \\
			\left( \frac{\partial \mathcal{Q}_m}{\partial x_s} \right)_{ms} & | & \left(\frac{\partial \mathcal{Q}_m}{\partial u_q} \right)_{mq}
		\end{pmatrix}.$$ 
		We shall now compute the partial derivatives, at the point $(x,r)$, as follows
		\begin{align*}
			\frac{\partial \mathcal{R}_j}{\partial x_s}  = \delta_{js} \sum_{i=1}^{n} \zeta_i \prod_{\ell \neq i} r_\ell^3 , \ \ \ \
			\frac{\partial \mathcal{R}_j}{\partial u_q}  = 3 u_q^2 \sum_{i=1}^n \zeta_i  \prod_{\ell \neq i,q} r_\ell^3 ,
		\end{align*}
		and 
		\begin{align*}
			\frac{\partial \mathcal{Q}_m}{\partial x_s}  = -2(x_r-p_{ms}) , \ \ \ \
			\frac{\partial \mathcal{Q}_m}{\partial u_q}  = 2r_m \delta_{mq} .
		\end{align*}
		Now, we must have that $r_\ell \neq 0$ for all $\ell \in \{ 1, \ldots , n \}$ as at an equilibrium $x$ cannot coincide with one of the charges $p_1, \ldots , p_n$. Therefore $\left(\frac{\partial \mathcal{Q}_m}{\partial u_q} \right)_{mq}$ has rank $n$ along the zero set of $\mathcal{P}$. On the other hand, $\left(\frac{\partial \mathcal{R}_j}{\partial x_s}\right)_{js}$ has rank $3$ provided that
		$$\sum_{i=1}^{n} \zeta_i \prod_{\ell \neq i} r_\ell^3 \neq 0.$$
		
		We shall now show that this is the case for the generic choice of $p_1 \in \mathbb{R}^3$. For that, it is enough to show that the map $F : \mathbb{R}^3 \to \mathbb{R}$ given by
		$$F(p_1)= \sum_{i=1}^{n} \zeta_i \prod_{\ell \neq i} r_\ell^3 , $$
		is a submersion at $p_1 \in \mathbb{R}^3$ such that $F(p_1)=0$ and $x \in \mathbb{R}^3$ as above. If the map $F$ is not a local submersion at such a point, all its partial derivatives must vanish. Then, by computing
		\begin{align*}
			\frac{\partial F}{\partial p_{1s}} & = \sum_{i=2}^n \zeta_i 2r_{1}^2 \frac{\partial r_1}{\partial p_{1s}} \prod_{\ell \neq 1,i} r_\ell^3 \\
			& = -2 \sum_{i=2}^n \zeta_i r_{1} (x_s-p_{1s}) \prod_{\ell \neq 1,i} r_\ell^3 \\
			& = -2 r_{1} (x_s-p_{1s}) \sum_{i=2}^n \zeta_i \prod_{\ell \neq 1,i} r_\ell^3 .
		\end{align*}
		we find that to be the case if and only if either
		$p_1=x$ or $\sum_{i=2}^n \zeta_i \prod_{\ell \neq 1,i} r_\ell^3 =0$. As mentioned before, the first condition is impossible as an equilibrium point cannot coincide with one of the charges. As for the second, together with $F(p_1)=0$, it would imply that both
		\begin{align*}
			\sum_{i=1}^n \zeta_i \prod_{\ell \neq i} r_\ell^3 & = 0 \\
			\sum_{i=2}^n \zeta_i \prod_{\ell \neq 1,i} r_\ell^3 & =0.
		\end{align*}
		The first equation above can also be written as
		$$\zeta_1 \prod_{\ell \neq 1} r_\ell^3 + r_1^3 \sum_{i=2}^n \zeta_i \prod_{\ell \neq 1,i} r_\ell^3 =0.$$
		Then, subtracting to it $r_1^3$ times the second equation gives
		$$\zeta_1 \prod_{\ell \neq 1} r_\ell^3  =0,$$
		which again, is impossible because an equilibrium point $x$ can not coincide with the location of a charge.
	\end{proof}

	\subsection{Finalizing the proof of Result \ref{thm:Main}}\label{ss:proof}
	
	By construction, and due to Lemma \ref{lem:Isolated zeroes v0}, the zeros of $\nabla V$ are in correspondence to isolated zeroes of 
	$$\mathcal{P}: \mathbb{R}^{3+n} \to \mathbb{R}^{3+n},$$
	for the generic choice of the point $p_1$. This is at most the number of isolated zeroes of its analytic extension $P: \mathbb{C}^{3+n} \to \mathbb{C}^{3+n}$ and Bezout's theorem \ref{thm:Bezout} guarantees that the number of such zeroes, counted with multiplicity, is the product of the degrees
	$$\deg(\mathcal{R}_1)\deg(\mathcal{R}_2)\deg(\mathcal{R}_3) \times \deg(\mathcal{Q}_1)  \ldots  \deg( \mathcal{Q}_n) .$$ 
	These degrees where computed in subsection \ref{sec:Polynomial} where we found this to be
	$$ (3n-2)^3 \times 2^n .$$

	\begin{remark}
		In \cite{Ki}, Bezout's theorem is also used to obtain an upper bound on the number of equilibria for planar charges. The procedure is, however very different. First, we consider any type of generically placed point charges, not just planar ones. Secondly, our procedure allows us to explicitly prove Lemma \ref{lem:Isolated zeroes v0}, which shows that the non-degenerate equilibria become isolated zeroes of the polynomial map $\mathcal{P}$. Thus allowing us to use Bezout's theorem to obtain an effective upper bound.
	\end{remark}

	\subsection{The modified potentials $V_p$}
	
	For $p>0$ let us instead consider the potential
	$$V_p(x)=\sum_{i=1}^n \frac{\zeta_i}{r_i^p}.$$
	
	\begin{proposition}\label{prop:Bound_p}
		For even $p=2r \in 2\mathbb{N}$, the number of isolated critical points of $V_p$ is at most 
		$$\left( r (n-1) +2-r \right)^3.$$ 
		In particular, the number of isolated critical points of $V_2$ is at most $n^3$.
	\end{proposition}
	
	\begin{proof}
		Differentiating with respect to $x_j$ and clearing out denominators, we obtain
		\begin{align*}
			\frac{\partial V_p}{\partial x_j} & = -\sum_{i=1}^n \frac{\zeta_i}{r_i^{p+2}} (x_j-p_{ij}) \\ 
			& = - \frac{1}{\prod_{\ell=1}^n r_\ell^{p+2}} \sum_{i=1}^n \zeta_i (x_j-p_{ij}) \prod_{m \neq i} r_m^{p+2}.
		\end{align*}
		If $p=2r \in 2\mathbb{N}$ is even, then the numerator above defines a polynomial in $x$
		$$\tilde{P_j}(x)=\sum_{i=1}^n \zeta_i (x_j-p_{ij}) \prod_{m \neq i} r_m^{2r+2}$$
		of degree 
		$$1+(n-1) (r-1) = r (n-1) +2-r.$$ 
		Again, by complexifying it and using the affine version of Bezout's theorem we conclude that
		$$\# \left( \tilde{P_1}^{-1}(0) \cap \tilde{P_2}^{-1}(0) \cap \tilde{P_3}^{-1}(0) \right) \leq \left( r (n-1) +2-r \right)^3 .$$
		In particular, for $r=1$ we obtain the upper bound $n^3$.
	\end{proof}
	
	\begin{remark}
		The bound obtained in Proposition \ref{prop:Bound_p} is an improvement of the previously known bound of
		$$(1+2(1+r)(n-1))\times (1+4(r+1)(n-1))^2,$$
		which had been obtained in \cite{Z}.
	\end{remark}

	\section{Regular Configurations}
	\label{sec:2}
	
	This section starts out by using Morse theory to show that for non-degenerate potentials generated by $n$ unit point charges in $\R^3$, there are at least $n-1$ equilibria.
	There are however configurations with more than $n-1$ equilibria, and as mentioned in the introduction, the quadratic upper bound $(n-1)^2$ was conjectured by Maxwell~\cite{Max54}. 
	This bound is realised for $n=3$, for example when placing charges at the vertices of an equilateral triangle. 
	However, configurations that achieve this bound are not known for $n>3$, and numerical experiments suggest that for randomly placed charges the ratio of the number of equilibria over point charges is bounded and rather low. 
	Therefore, we shall now study special configurations of point charges, such as those obtained by placing them at the vertices of certain solids. 
	This study leads in Section~\ref{sec:2.3} to our first main result.
	
	\subsection{Morse Theory Considerations}
	\label{sec:2.1}
	
	An equilibrium of a smooth function is \emph{non-degenerate} if the Hessian at the point is invertible.
	A useful result of Morse--Cairns in \cite{MoCa69} yields that for the generic location of a single point charge, all equilibria of $V$ are non-degenerate.
	In this case, the Hessian has $0$, $1$, $2$, or $3$ negative eigenvalues, and in Morse theory, this number is called the \emph{index}; see e.g.\ \cite{Mil63}.
	Equilibria of index $0, 1, 2, 3$ are referred to as \emph{minima}, \emph{$1$-saddles}, \emph{$2$-saddles}, and \emph{maxima}.
	Importantly, the electrostatic potential in $\R^3$ is \emph{harmonic}; that is: its Laplacian vanishes at every point.
	To explain, it is not difficult to check that the Laplacian of the function $x \mapsto \sfrac{1}{\norm{x}}$ is harmonic on $\R^3 \backslash \lbrace 0 \rbrace$. Sums, translates, and multiples of harmonic functions are again harmonic, which implies that $V$ is harmonic.
	It follows that $V$ enjoys the \emph{mean value property}: $V(x)$ is the average of the values of $V$ on any sphere centered at $x$ that neither passes through nor encloses any of the $A_i$.
	Hence, if nonzero, $V$ has neither minima nor maxima in $\R^3 \setminus \{A_1, A_2, \ldots, A_n\}$.
	
	\smallskip
	An early application of Morse theory---carried out by Morse himself \cite{MoCa69}---gives a lower bound on the number of equilibria of $V$. 
	We shall now describe this argument, and its outcome under the assumption that $\zeta_i > 0$ for $1 \leq i \leq n$ and that all the equilibria are non-degenerate (as mentioned above, this can be achieved by simply perturbing one of the point charges). 
	Because $\R^3 \setminus \{A_1, A_2, \ldots, A_n\}$ is not compact, it is convenient to modify $V$ to get a Morse function, $M \colon \Sp^3 \to \R$, that is sufficiently similar to $V$ as follows:
	\begin{enumerate}
		\item compactify $\R^3$ to $\Sp^3$ by adding a point at infinity and extend $V$ by continuity, so it takes the value $0$ and has a local minimum there;
		\item for each $1 \leq i \leq n$, define $M$ inside a sufficiently small neighbourhood of $A_i$ such that $A_i$ is a maximum and the only equilibrium of $M$ in this neighbourhood;
		\item let $M$ coincide with $V$ outside these neighbourhoods and at the point at infinity.
	\end{enumerate}
	Let $m_p$ be the number of index-$p$ equilibria of $M$.
	After subtracting the one minimum at infinity and the $n$ maxima at $A_1$ to $A_n$, we get $m_0 + m_1 + m_2 + m_3 - (n+1)$ as the number of equilibria of $V$.
	By the Euler--Poincar\'{e} relation, we have $m_0 - m_1 + m_2 - m_3 = 0$ because the Euler characteristic of $\Sp^3$ vanishes.
	By construction, $m_3 = n$ and $m_0 = 1$, which thus implies
	$m_2 - m_1 = n-1$.
	Hence, the number of equilibria of $V$ is
	\begin{align}
		m_1+m_2 &= 2m_1 + n-1 \geq n-1 .
	\end{align}
	This lower bound is tight since placing $n$ point charges along a straight line in $\R^3$ defines a potential with only $n-1$ equilibria, all of which are $2$-saddles.
	
	\subsection{Platonic Solids and Local Homology}
	\label{sec:2.2}
	
	We consider the vertices of the Platonic solids and semi-regular polyhedra as potentially interesting point charge configurations.
	The numbers of equilibria given in Table~\ref{tbl:solids} were first obtained by numerical computations and graphical investigations, in which we discovered that all equilibria are non-degenerate, except the center of the solid, which unfolds into two or more such equilibria.
	Through lengthy but straightforward computations and arguments, we were then able to rigorously show the existence of the claimed equilibria. 
	As an example, we present the full proof for the case of the cube in Appendix~\ref{app:A}. 
	
	\begin{table}[hbt]
		\centering \footnotesize{
			\begin{tabular}{l||c|ccc}
				\multicolumn{1}{c||}{solid} & $f$-vector     & 2-saddles & 1-saddles & center  \\ \hline \hline
				Tetrahedron         & (4,6,4,1)    &  0 &  0 &   3    \\
				Cube                & (8,12,6,1)   & 12 &  0 &  -5    \\
				Octahedron          & (6,12,8,1)   &  0 &  0 &   5    \\
				Dodecahedron        & (20,30,12,1) & 30 &  0 & -11    \\
				Icosahedron         & (12,30,20,1) &  0 &  0 &  11
		\end{tabular} }
		\vspace{0.1in}
		\caption{\footnotesize  Numerical results for the electrostatic potential 
			defined by placing the point charges at the vertices of the five Platonic solids.
			The first entry in each $f$-vector is the number of point charges.
			The next two columns give the number of observed $2$- and $1$-saddles.
			In each case, the center is a degenerate equilibrium, for which we give the alternating sum of ranks in local homology.
			In each case, the number vertices exceeds this alternating sum plus the number of $2$-saddles minus the number of $1$-saddles by $1$.}
		\label{tbl:solids}
	\end{table}
	
	\smallskip
	We use the local homology at a point to determine whether it is an equilibrium or not.
	For $x \in \R^3 \setminus \{A_0, A_1, \ldots, A_n\}$, let $\ee > 0 $ be smaller than the distance between $x$ and any of the $n$ point charges, and write $B(x,\ee)$ and $S(x,\ee)$ for the closed ball and sphere with center $x$ and radius $\ee$.
	The \emph{local homology} of $V$ at $x$ is the limit, for $\ee$ going to zero, of the relative homology of the pair $\left( B^-(x,\ee), S^-(x,\ee) \right)$, in which $B^-(x,\ee) \subseteq B(x,\ee)$ and $S^-(x,\ee) \subseteq S(x,\ee)$ are the subsets of points $y$ with $V(y) \leq V(x)$.
	To visually represent these groups, we define a binary function on the unit sphere, which maps $u \in \Sp^2$ to \emph{black} if $V(x) < V(x+\ee u)$, and to \emph{white} if $V(x) \geq V(x+\ee u)$ for every sufficiently small $\ee > 0$.
	In words, the potential increases if we leave $x$ in a black direction, and it decreases or stays the same if we leave $x$ in a white direction; see the upper row in Figure~\ref{fig:binary_functions} for the binary functions of a non-critical point and the four types of non-degenerate equilibria.
	\begin{figure}[hbt]
		\centering
		\vspace{-0.3in}
		\hspace{-0.4in}
		\includegraphics[width=.28\linewidth]{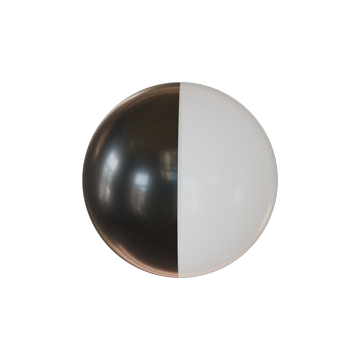}
		\hspace{-0.7in}
		\includegraphics[width=.28\linewidth]{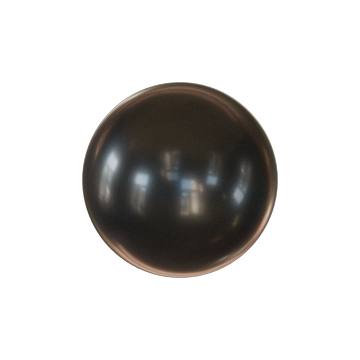}
		\hspace{-0.7in}
		\includegraphics[width=.28\linewidth]{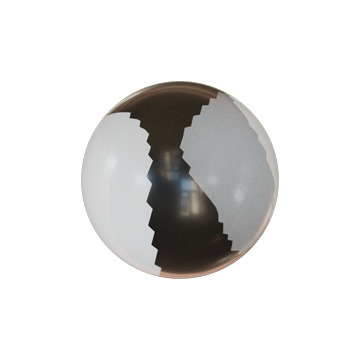}
		\hspace{-0.7in}
		\includegraphics[width=.28\linewidth]{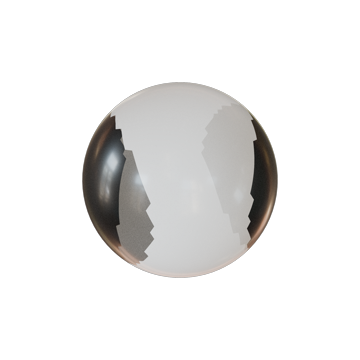}
		\hspace{-0.7in}
		\includegraphics[width=.28\linewidth]{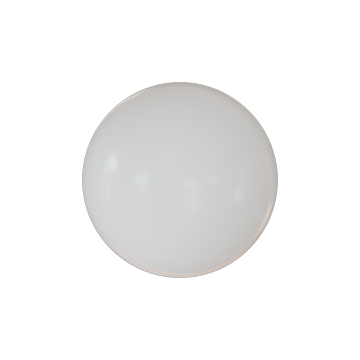}
		\hspace{-0.4in}
		\vspace{-0.6in} \\
		\hspace{-0.4in}
		\includegraphics[width=.28\linewidth]{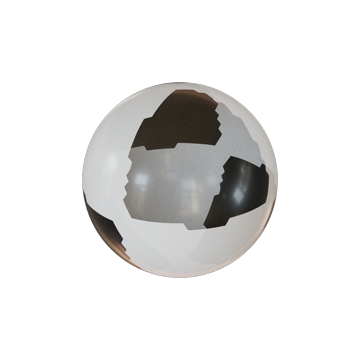}
		\hspace{-0.7in}
		\includegraphics[width=.28\linewidth]{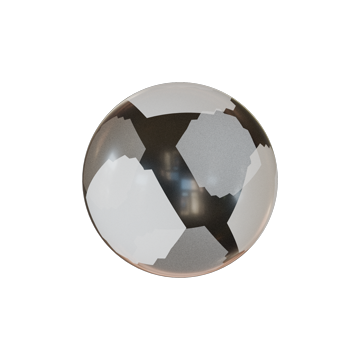}
		\hspace{-0.7in}
		\includegraphics[width=.28\linewidth]{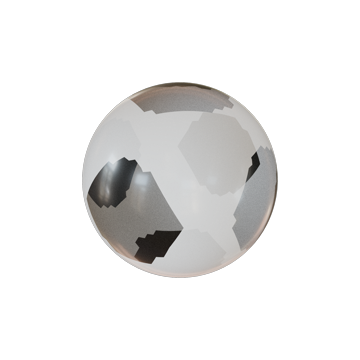}
		\hspace{-0.7in}
		\includegraphics[width=.28\linewidth]{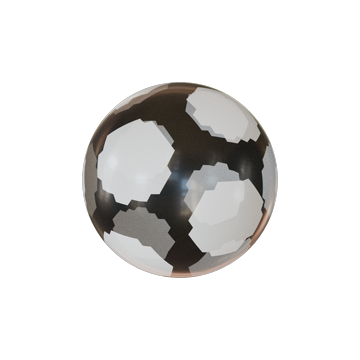}
		\hspace{-0.7in}
		\includegraphics[width=.28\linewidth]{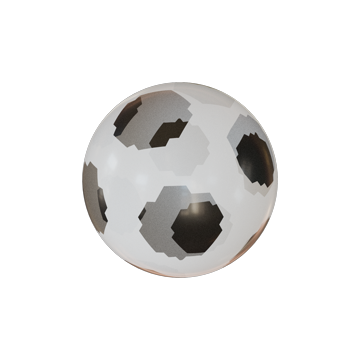}
		\hspace{-0.4in}
		\vspace{-0.3in}
		\caption{\footnotesize \emph{Upper row, from left to right:} the binary functions on the unit $2$-sphere for a non-critical point, a minimum, a $1$-saddle, a $2$-saddle, and a maximum.
			\emph{Lower row, from left to right:} the binary functions for the centers of the tetrahedron, cube, octahedron, dodecahedron, and isocahedron (these are degenerate equilibria and so are neither $1$-saddles nor $2$-saddles).}
		\label{fig:binary_functions}
	\end{figure}
	Writing $W \subseteq \Sp^2$ for the white points, the local homology of $x$ is the homology of the cone over $W$ relative to $W$; see the left and right halves of Table~\ref{tbl:ranks} for the ranks of these relative groups for the cases displayed in the first and second rows of Figure~\ref{fig:binary_functions}, respectively.
	For example, the rank vector of the tetrahedron is $(0,0,3,0)$, which is consistent with three $2$-saddles co-located at its center.
	It is also consistent with one minimum, four $1$-saddles, and six $2$-saddles co-located at the center, where the four $1$-saddles cancel with the minimum and three of the $2$-saddles.
	The latter choice is suggested by the face structure of the tetrahedron, which has one tetrahedron, four triangles, and six edges.
	\begin{table}[hbt]
		\centering \footnotesize {
			\begin{tabular}{l|rrrr||l|rrrr}
				&$p=0$& 1 & 2 & 3 &             &$p=0$& 1 &  2 &  3  \\  \hline \hline
				non-critical & 0 & 0 & 0 & 0 &  tetrahedron  & 0 &  0 &  3 & 0 \\
				minimum      & 1 & 0 & 0 & 0 &  cube         & 0 &  5 &  0 & 0 \\
				$1$-saddle   & 0 & 1 & 0 & 0 &  octahedron   & 0 &  0 &  5 & 0 \\
				$2$-saddle   & 0 & 0 & 1 & 0 &  dodecahedron & 0 & 11 &  0 & 0 \\
				maximum      & 0 & 0 & 0 & 1 &  icosahedron  & 0 &  0 & 11 & 0 \\
		\end{tabular} }
		\vspace{0.1in}
		\caption{\footnotesize  \emph{Left half:} the ranks in local homology at the points whose binary functions are illustrated in the upper row of Figure~\ref{fig:binary_functions}.
			\emph{Right half:} the ranks in local homology at the centers of the Platonic solids, with binary functions shown in the lower row of Figure~\ref{fig:binary_functions}.
			Observe that the last column in Table~\ref{tbl:solids} is the alternating sum of these ranks.}
		\label{tbl:ranks}
	\end{table}
	
	\smallskip
	Observe that every geometric realization of the solid yields the same local homology. 
	Indeed, a \emph{similarity}---which is a map $\theta \colon \R^3 \to \R^3$ that is the composition of a scaling, a rotation, a translation, and possibly a reflection---preserves the types of equilibria.
	To prove this claim, we write $x' = \theta (x)$ and $V' \colon \R^3 \setminus \{A_1', A_2', \ldots, A_n'\} \to \R$ for the potential defined by the point $A_i'$ with charges $\zeta_i' = \zeta_i$.
	\begin{lemma}
		\label{lem:similarity}
		Let $V, V'$ be the potentials defined by a finite set of point charges, before and after applying a similarity, respectively.
		Then there is a bijection between the equilibria of $V$ and $V'$
		such that corresponding equilibria have the same type.
	\end{lemma}
	\begin{proof}
		A rotation, translation, and reflection does not change the value at a point, except that it transfers this value to the image of that point.
		This is not true for a scaling.
		Letting $s > 0$ be the scaling factor of the similarity, we have
		\begin{align}
			V'(x') &= \sum\nolimits_i \frac{\zeta_i'}{\dist{x'}{A_i'}}
			= \sum\nolimits_i \frac{\zeta_i}{s\dist{x}{A_i}}
			= \tfrac{1}{s} V(x) ,
		\end{align}
		at every $x \in \R^3$.
		Since multiplication with a positive constant preserves the local homology at a point, this establishes the claim for the bijection that maps an equilibrium to its image under the similarity.
	\end{proof}

	\subsection{Prisms and Anti-prism}
	\label{sec:2.3}
	
	Recall that the unit point charges at the vertices of the equilateral triangle in $\R^3$ define four equilibria, which implies that the ratio of equilibria over the number of vertices is $\sfrac{4}{3}$.
	Based on the numbers given in Table~\ref{tbl:solids}, we deduce that the corresponding ratios for the Platonic solids are $\sfrac{1}{4}$, $\sfrac{13}{8}$, $\sfrac{1}{6}$, $\sfrac{31}{20}$, and $\sfrac{1}{12}$, respectively, with the highest ratio for the cube (see also Appendix \ref{app:A} for a proof of the existence of these equilibria).
	We get an even higher ratio of $\sfrac{33}{12}$ for the cuboctahedron; see the second row of Table~\ref{tbl:Archimedean-solids}.
	Indeed, this is the highest ratio we observe for the regular and semi-regular polytopes surveyed in Section~\ref{sec:2.2} and Appendix~\ref{app:B}.
	
	\smallskip
	Beyond these families, we wish to consider the prisms and anti-prisms, which are parametrized by an integer, $k$, and a positive real number.
	A \emph{$k$-sided prism} is the Cartesian product of a regular $k$-gon and an interval.
	It has two regular $k$-gons and $k$ rectangles as facets.
	We call the apect ratio of the rectangles the \emph{relative height} of the prism.
	Lemma~\ref{lem:similarity} applies only if two prisms agree on the $k$ and on the relative height.
	Indeed, for the same $k$, the number of equilibria generated by unit point charges at the vertices of the prism depends on the relative height.
	We shed light on this dependence for the more interesting second family.
	A \emph{$k$-sided anti-prism} has two regular $k$-gons in parallel planes as facets, one rotated relative to the other so that there are $2k$ isosceles triangles that connect the $k$-gons and complete the list of boundary facets; see Figure~\ref{fig:anti-prism} for examples.
	The \emph{relative height} of the anti-prism is the distance between the two $k$-gons over the length of their edges.
	Note that for $k=3$ and the relative height chosen so that the isosceles triangles are equilateral, the $3$-sided or triangular anti-prism is the (regular) octahedron.
	
	\smallskip
	Our experiments show that for $k \geq 4$, the relative height of the $k$-sided anti-prism can be chosen so that the electrostatic potential for unit point charges placed at the vertices has one $2$-saddle for each edge, one $1$-saddle for each isosceles triangle, and one additional $1$-saddle at the center; see Figure~\ref{fig:anti-prism} for illustrations and Appendix~\ref{app:C} for analytic proofs of some of these equilibria.
	For $k \geq 4$, the ratio of equilibria over vertices is therefore
	\begin{align}
		\frac{1}{2k} (4k + 2k + 1) &= 3 + \frac{1}{2k} ,
	\end{align}
	For $k = 6$, this is the same ratio we get for the cuboctahedron after perturbing the point charges so that the degenerate equilibrium at the center unfolds into individual non-degenerate equilibria; see Section~\ref{sec:2.4} for details.
	To maximize this ratio, we minimize the number of sides, which suggests that the $4$-sided or square anti-prism is our best choice.
	In contrast, the equilibria for the $3$-sided or triangular anti-prism are always too close to distinguish using our methods.
	\begin{figure}[htb]
		\centering \vspace{-0.5in}
		\hspace{-1.1in}
		\includegraphics[width=.45\linewidth]{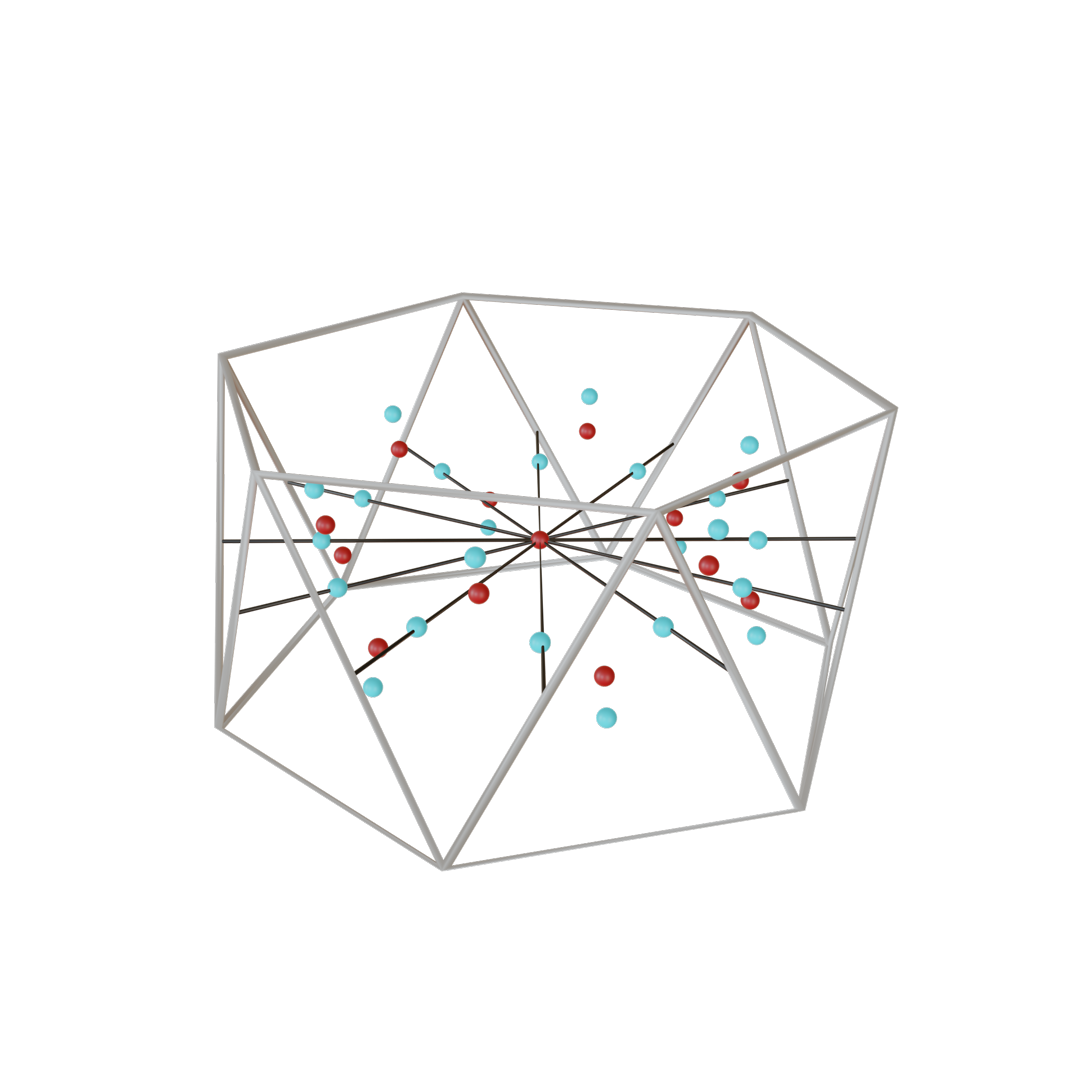} \hspace{-0.65in}
		\includegraphics[width=.45\linewidth]{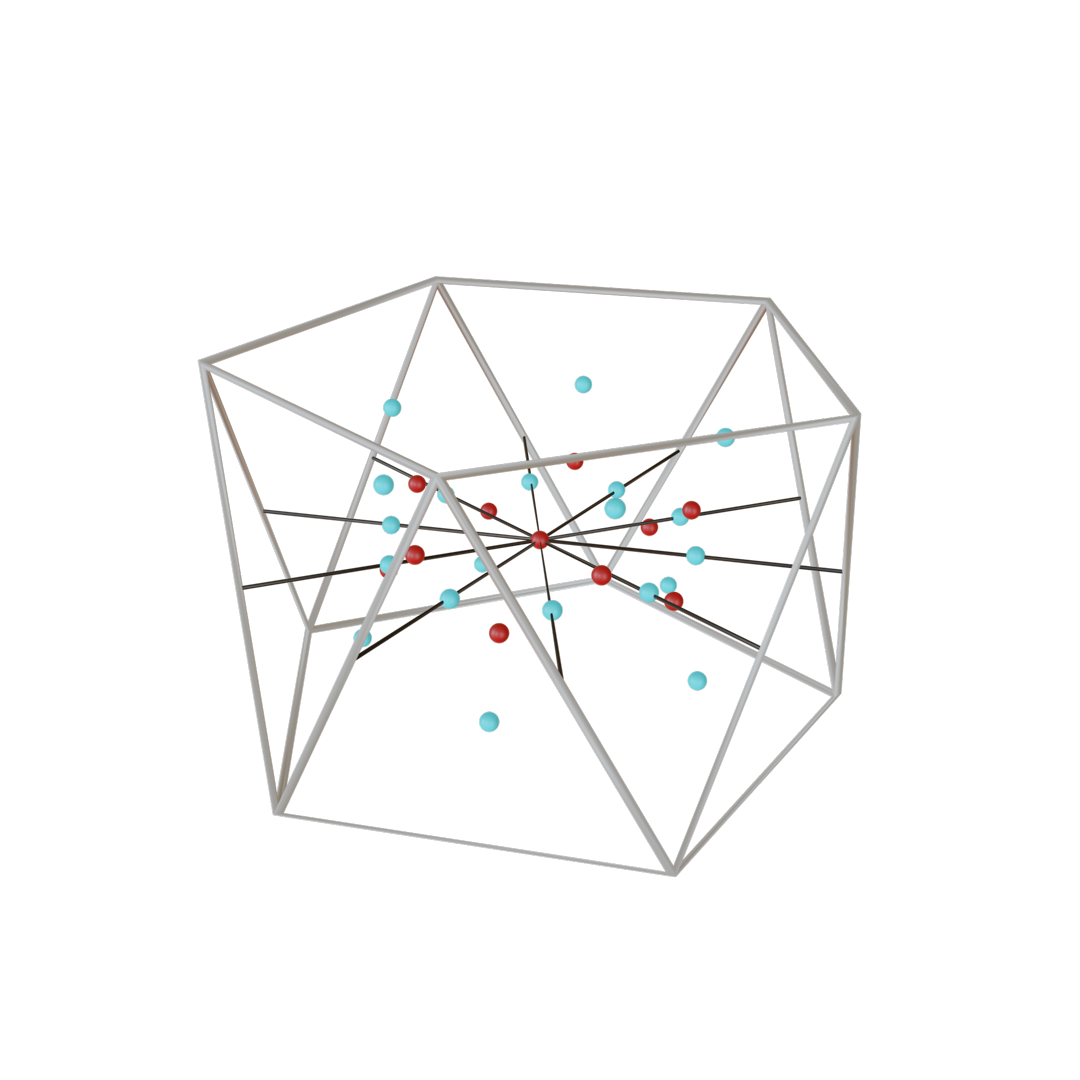} \hspace{-0.8in} 
		\includegraphics[width=.45\linewidth]{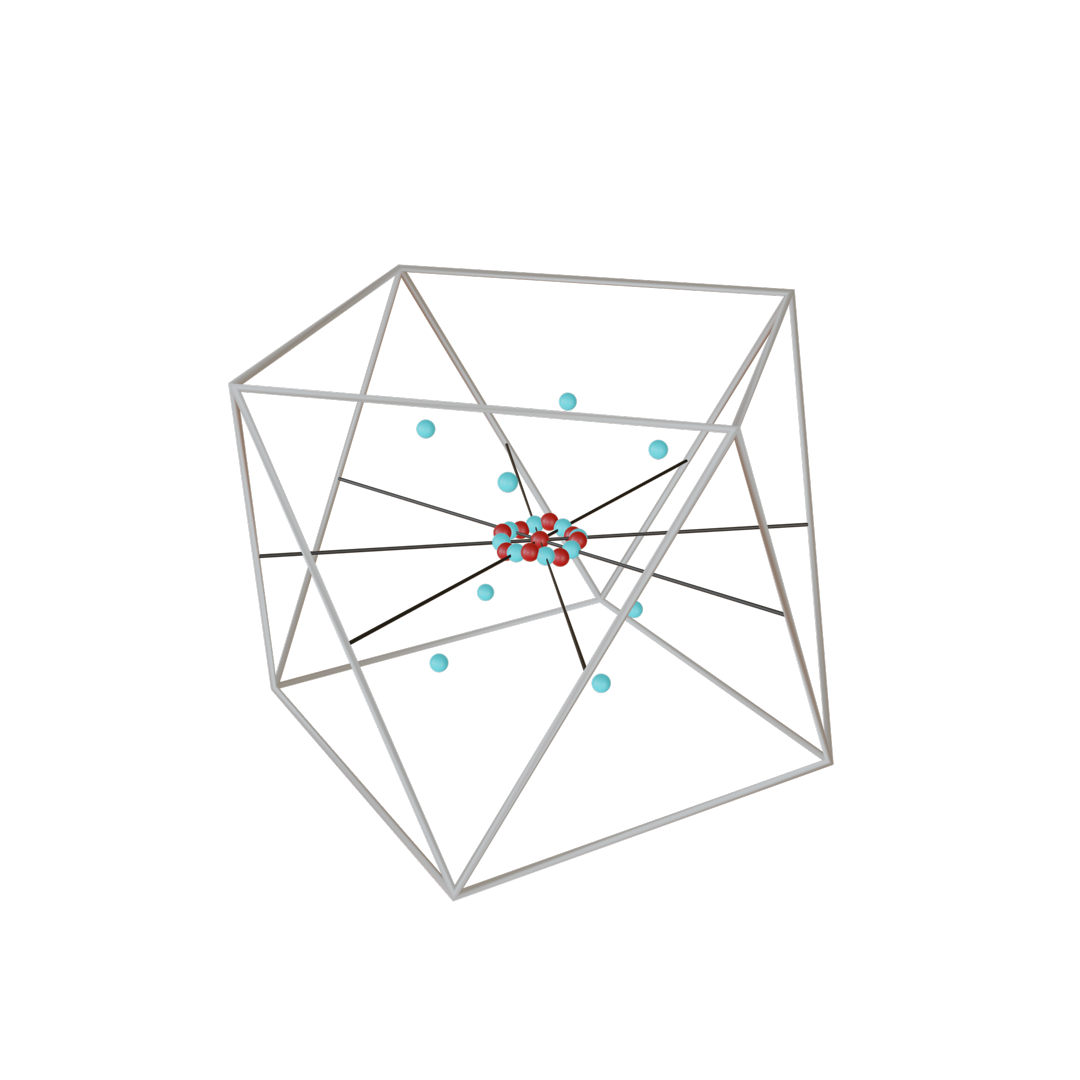} \hspace{-1.2in} 
		\vspace{-0.35in}
		\caption{\footnotesize \emph{From left to right:} the hexagonal, pentagonal, square anti-prisms with the heights chosen to maximize the number of equilibria.
			The ratios of equilibria over vertices are $\sfrac{37}{12} < \sfrac{31}{10} < \sfrac{25}{8}$, respectively.
			Observe how a ring of alternating $1$- and $2$-saddles gets successively more concentrated around the center.}
		\label{fig:anti-prism}
	\end{figure}
	
	\smallskip
	In addition to the results on the maximum number of equilibria for anti-prims, we discover an interesting universal phenomenon: the existence of a transition that occurs when the relative height of the prism or anti-prism is $\sqrt{2}$. 
	This transition is universal in the sense that it occurs for both prisms and anti-prims formed with any regular $k$-gon and always at relative height $\sqrt{2}$. 
	We prove this in Appendix~\ref{app:C} and describe it here together with an existence result for some other equilibria.
	
	\begin{theorem}
		\label{thm:prisms}
		Consider the electrostatic potential generated by unit point charges at the vertices of any prism or anti-prism. 
		\begin{enumerate}
			\item Its center is a non-degenerate equilibrium iff the relative height is not $\sqrt{2}$, in which case the center is a $1$-saddle for relative heights less than $\sqrt{2}$ and a $2$-saddle for relative heights larger than $\sqrt{2}$. 
			\item Moreover, if the relative height is larger than $\sqrt{2}$, there are two additional $1$-saddles along the axis of the prism or anti-prism.
			\item Furthermore, if the relative height is smaller than $\sqrt{2}$, there is an equilibrium on every line segment connecting the center of the anti-prism to the midpoint of a lateral edge.
		\end{enumerate}
	\end{theorem}
	
	\Skip{\begin{theorem}
			Consider the electrostatic potential generated by unit point charges at the vertices of any prism or anti-prism. Then, its center is a non-degenerate equilibrium iff the relative height is not $\sqrt{2}$, in which case the center is a $1$-saddle for relative heights less than $\sqrt{2}$ and a $2$-saddle for relative heights larger than $\sqrt{2}$. Moreover, in the latter case there are two additional $1$-saddles along the axis of the prism or anti-prism respectively.
		\end{theorem}
		
		The relative height $\sqrt{2}$ is also special in one extra result which we prove in Appendices \ref{app:C.3} and \ref{app:C.4}. This locates some equilibria for the triangular and square anti-prism which exist for relative heights below the critical value of $\sqrt{2}$.
		
		\begin{theorem}
			Consider the electrostatic potential generated by unit point charges at the vertices of a triangular or square anti-prism. Then, for relative heights smaller than $\sqrt{2}$, there is an equilibrium on any of the segments connecting the center of the anti-prism to the midpoints of every lateral edge.
		\end{theorem}
	} 

	\subsection{Increasing the Ratio}
	\label{sec:2.4}
	
	We can further increase the ratio in two steps. 
	The first of these arises from the realization that the equilibrium at the center of a Platonic solid is necessarily degenerate and we can perturb the location of the charges so that it unfolds into several non-degenerate ones. 
	Indeed, by Theorem~6.3 in \cite{MoCa69}, the point charges can be perturbed such that all equilibria are non-degenerate and arbitrarily close to the degenerate equilibria that give rise to them. 
	Setting $N = m_1 - m_2 + n - 1$, we note that this is the alternating sum of local homology ranks at the center.
	It is also a lower bound for the number of $1$-saddles and $2$-saddles the degenerate equilibrium unfolds into.
	For example, we get three $2$-saddles near the center of the tetrahedron, five $1$-saddles near the center of the cube, and so on.
	Perturbing the vertices of the five Platonic solid thus increases the ratios to $\sfrac{3}{4}$, $\sfrac{17}{8}$, $\sfrac{5}{6}$, $\sfrac{41}{20}$, and $\sfrac{11}{12}$, respectively.
	Perturbing the vertices of the cuboctahedron increases the ratio to $\sfrac{37}{12}$.
	Compare this with the ratios $\sfrac{37}{12} < \sfrac{31}{10} < \sfrac{25}{8}$ for the hexagonal, pentagonal, square anti-prisms displayed in Figure~\ref{fig:anti-prism}, which are obtained without any perturbation.
	
	\smallskip
	In the second step, we iteratively substitute the vertices of a much smaller copy of the solid for each vertex of the original solid.
	Iterating $\ell - 1$ times, we get a configuration with $\ell$ layers.
	Letting $n$ be the number of vertices of the solid and $m$ the number of equilibria defined by these $n$ charges, the number of point charges in the $\ell$ layers is $n_\ell = n^\ell$, and the number of equilibria is
	\begin{align}
		m \cdot \left( n^{\ell-1} + n^{\ell-2} + \ldots + 1 \right) &= m \cdot \frac{n^{\ell}-1}{n-1}
		= \frac{m}{n-1} \cdot (n_\ell-1) .
	\end{align}
	In particular, the ratio of the number of equilibria by electric charges in the $\ell$-th iteration is given by
	$\frac{m}{n-1} \cdot \frac{n_\ell-1}{n_\ell}$,
	which for $\ell=1$ coincides with the initial ratio of $\frac{m}{n}$, and as $\ell \to + \infty$ converges to $\frac{m}{n-1}$.
	In all examples we have computed, the iterated square anti-prism achieves the largest fraction $\sfrac{25}{7}$ in the limit, when $\ell \to \infty$.
	This is the highest ratio ever observed by far.
	In particular, we find that for any $\ee>0$, there is an iteration $\ell_*$ such that all further iterations $\ell > \ell_*$ have a ratio greater than $\sfrac{25}{7} - \ee$. 
	This establishes Result~\ref{result:Antiprism} and naturally leads to the question of whether one can find an example of a configuration which achieves a greater ratio. 
	We state this inquiry as follows.
	
	\medskip \noindent \textbf{Question.}
	Can one find an example of a configuration of $n$ unit point changes in $\R^3$ whose electrostatic potential has a number of equilibria that exceeds $\sfrac{25 n}{7}$?

	\section{1-parameter Family of Potentials}
	\label{sec:3}
	
	In the same way as in \cite{GNS07}, we generalize the set-up by introducing a real parameter, $p > 0$, which modifies the effect of the distance to the point charges on the potential function:
	\begin{align}
		V_p (x) &= \sum\nolimits_{i=1}^n \left( \frac{\zeta_i}{\dist{x}{A_i}} \right)^p .
		\label{eqn:Vp}
	\end{align}
	Comparing \eqref{eqn:Vp} with \eqref{eqn:V}, we see that $V = V_{1}$.
	The main purpose of this $1$-parameter family of functions is to interpolate between the electrostatic potential and the (weighted) Euclidean distance function.
	\begin{figure}[htb]
		\centering
		\hspace{-1.1in}
		\includegraphics[width=.33\linewidth]{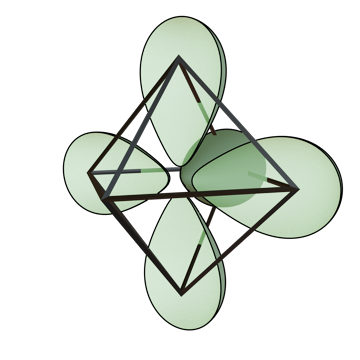} \hspace{-0.4in} 
		\includegraphics[width=.33\linewidth]{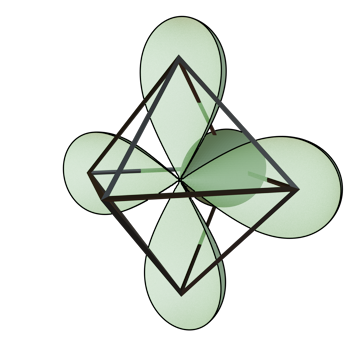} \hspace{-0.4in}
		\includegraphics[width=.33\linewidth]{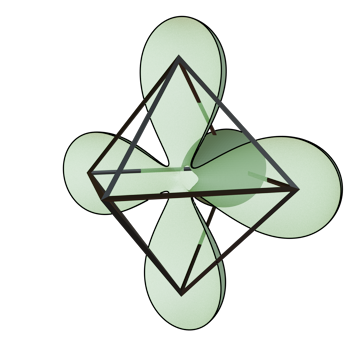}  \\
		\vspace{-0.6in}
		\hspace{0.3in}
		\includegraphics[width=.33\linewidth]{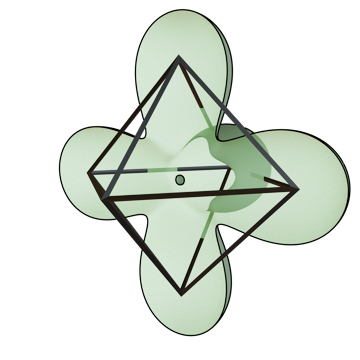} \hspace{-0.4in}
		\includegraphics[width=.33\linewidth]{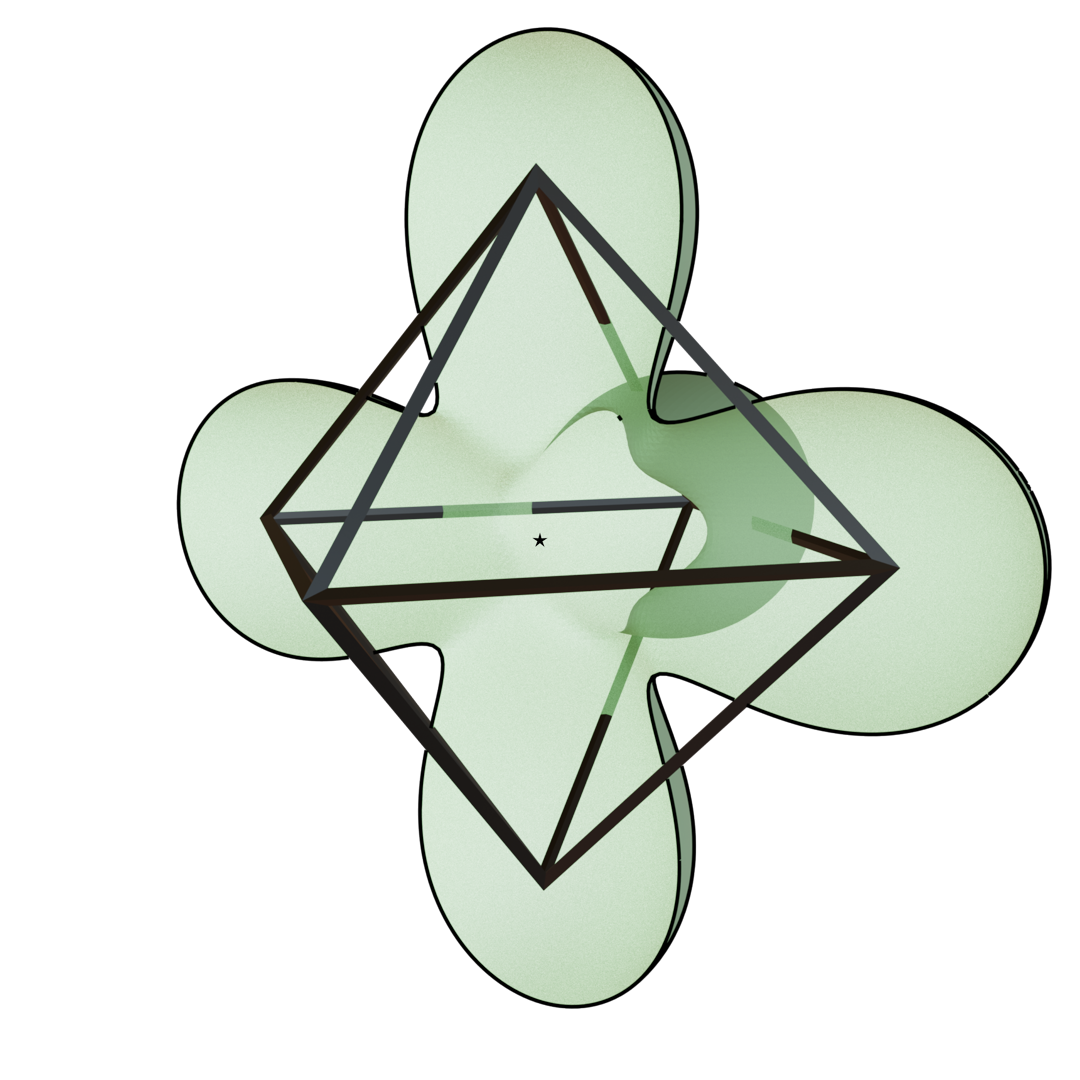} \hspace{-0.4in}
		\includegraphics[width=.33\linewidth]{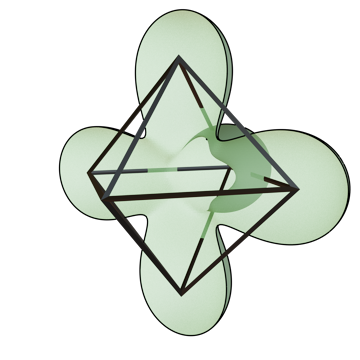} 
		\vspace{-0.1in}
		\caption{\footnotesize Cut-away views of three level sets of $V_{1}$ (\emph{upper row}) and three level sets of $V_{1.3}$ (\emph{lower row}) defined by point sources at the vertices of the octahedron.
			\emph{From left to right:} the values are chosen slightly less than, equal to, and slightly greater than the potential at the center of the octahedron.
			Removing the front of the surface reveals some of the complication at the center, which for $V_{1}$ is a degenerate equilibrium; compare with Figure~\ref{fig:binary_functions}, but for $V_{1.3}$ is a minimum with a single point in the level set at the center.}
		\label{fig:Octahedron-level-sets}
	\end{figure}
	The latter is the limit of one over the $p$-th root of $V_p$, as $p$ goes to infinity:
	\begin{align}
		E(x) &= \lim_{p \to \infty} \frac{1}{\sqrt[p]{V_p (x)}}
		= \lim_{p \to \infty} \left( \sum\nolimits_{i=1}^n \frac{\zeta_i^p}{\dist{x}{A_i}^p} \right)^{-1/p}
		= \min_{1 \leq i \leq n} \frac{\dist{x}{A_i}}{\zeta_i} .
		\label{eqn:Edistance}
	\end{align}
	
	\subsection{No Maxima}
	\label{sec:3.1}
	
	For $p=1$, $V_p$ is harmonic and therefore has neither minima nor maxima.
	A weaker property still holds for more general values of $p$.
	Specifically, if $p \geq 1$ and all charges are positive, then $V_p$ does not have any maxima:
	\begin{proposition}
		\label{prop:no_maxima}
		Let $V_p \colon \R^3 \setminus \{A_1,A_2,\ldots,A_n\} \to \R$ as defined in \eqref{eqn:Vp}, and assume $\zeta_i > 0$ for every $1 \leq i \leq n$.
		Then for every $p \geq 1$, $V_p$ has no maximum.
	\end{proposition}
	We present the proof in Appendix~\ref{app:D} and illustrate the difference between the cases $p=1$ and $p>1$ in Figure~\ref{fig:Octahedron-level-sets}, which shows a few level sets of $V_{1}$ and $V_{1.3}$.
	The tiny sphere around the center of the octahedron in the lower left level set suggests that the center is a genuine minimum of $V_{1.3}$, which it cannot be if $p = 1$ since this would contradict the harmonicity of $V_{1}$.
	In fact, we can prove the following result showing that for $p>1$, the origin is indeed a local minimum of the potential.
	\begin{proposition}\label{prop:octahedron}
		Let $\{ A_1, A_2, \ldots, A_6 \}$ be the vertices of an octahedron and for $p\geq 1$
		$$V_p(x)=\sum_{i=1}^6 \frac{1}{\norm{x-A_i}^p}.$$ 
		Then, the center of the octahedron is a local minimum of $V_p$ for all $p>1$.
	\end{proposition}
	Again, the proof of this result will be presented in Appendix \ref{app:D}. Finally, we point out that in this and other examples, the origin becomes a non-degenerate equilibrium for $p>1$. 
	Thus, in terms of non-degeneracy, increasing the power, $p$, can have a similar effect as perturbing the point charges; see \cite[Thm~6.3]{MoCa69}.

	\subsection{Counterexample}
	\label{sec:3.2}
	
	For positive and equal charges, the function $E \colon \R^3 \to \R$ defined in \eqref{eqn:Edistance} is commonly referred to as the \emph{Euclidean distance function} of the $A_i$ and studied through the \emph{Voronoi tessellation}, which assigns to each $A_i$ the region of points that are at least as close to $A_i$ as to any of the other points.
	If the charges are not necessarily the same (but still positive), then $E \colon \R^3 \to \R$ is a weighted version of the Euclidean distance function, namely with multiplicative weights ${1}/{\zeta_i}$.
	The corresponding tessellation of $\R^3$ is the \emph{multiplicatively weighted Voronoi tessellation} \cite{AuEd84,OBSC00}.
	Both types of tessellations are described in more detail in Section~\ref{sec:4}.
	
	\smallskip
	The hope expressed as Conjecture~1.8 in \cite{GNS07}---and stated as Conjecture~\ref{conj:1.8} in this article---is that $V \colon \R^3 \to \R$ cannot have more equilibria than $E \colon \R^3 \to \R$, in which the latter are defined using the limit process in \eqref{eqn:Edistance}.
	As stated in Result~\ref{result:Counterexample}, the unit point charges at the vertices of the truncated octahedron form a counterexample to this conjecture, and we shall now explain why.
	This Voronoi domain is also known as the truncated octahedron; see Figure~\ref{fig:BCC-counterexample}, which also shows the eqilibria of the electrostatic potential defined by placing unit point charges at its vertices. 
	This polytope has $14$ faces ($6$ squares and $8$ hexagons) and $36$ edges ($24$ shared by a square and a hexagon and $12$ shared by two hexagons).
	Correspondingly, $E$ has $14$ $1$-saddles and $36$ $2$-saddles, which we compare to the $18$ $1$-saddles and $36$ $2$-saddles of $V$.
	Hence, the Euclidean distance function has fewer $1$-saddles than the electrostatic potential of the same unit point charges, and it has equally many $2$-saddles.
	In total, $V$ has more equilibria than $E$, which contradicts Conjecture~1.8 in \cite{GNS07}.
	\begin{figure}[hbt]
		\centering
		\includegraphics[width=.42\textwidth]{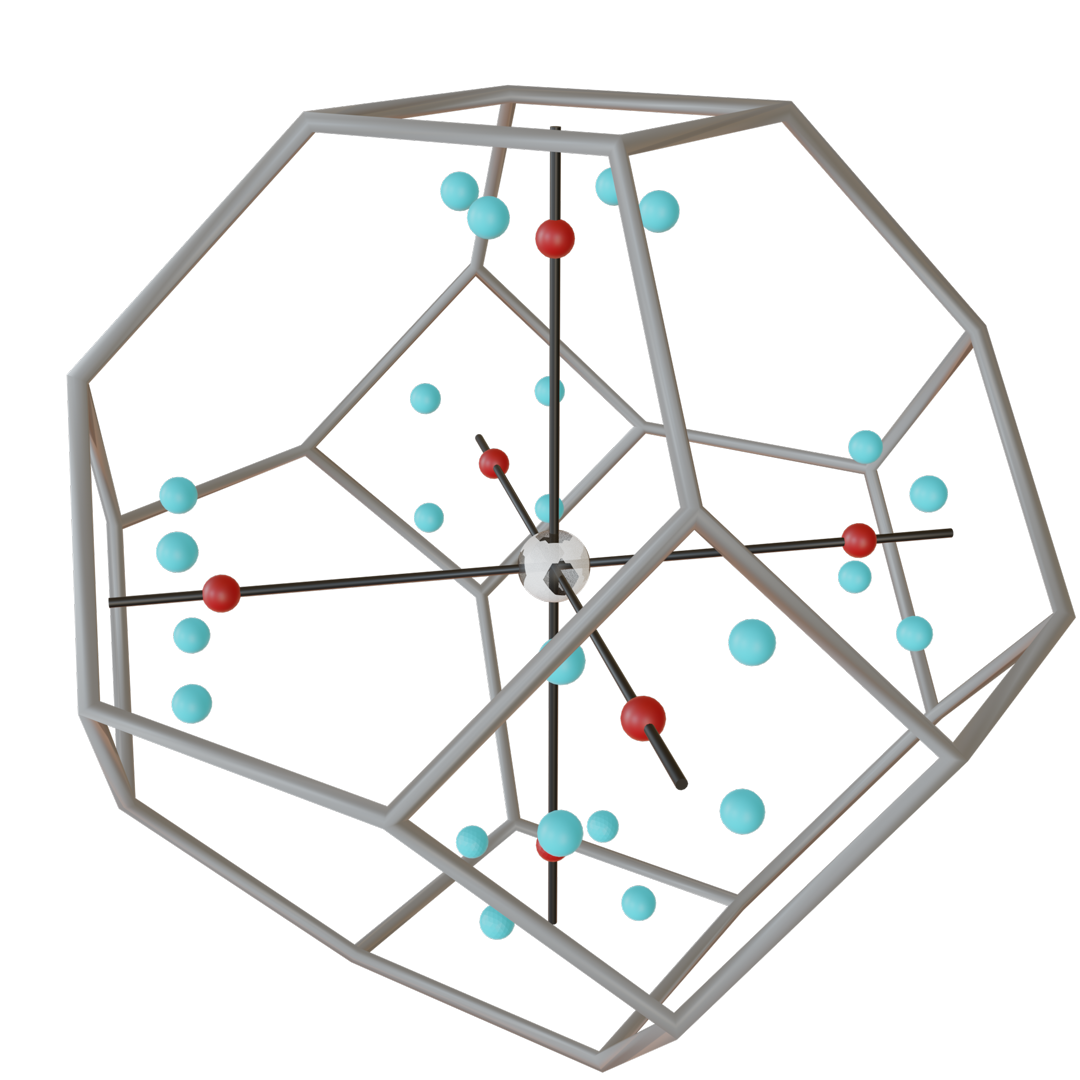} \hspace{0.2in} 
		\includegraphics[width=.42\textwidth]{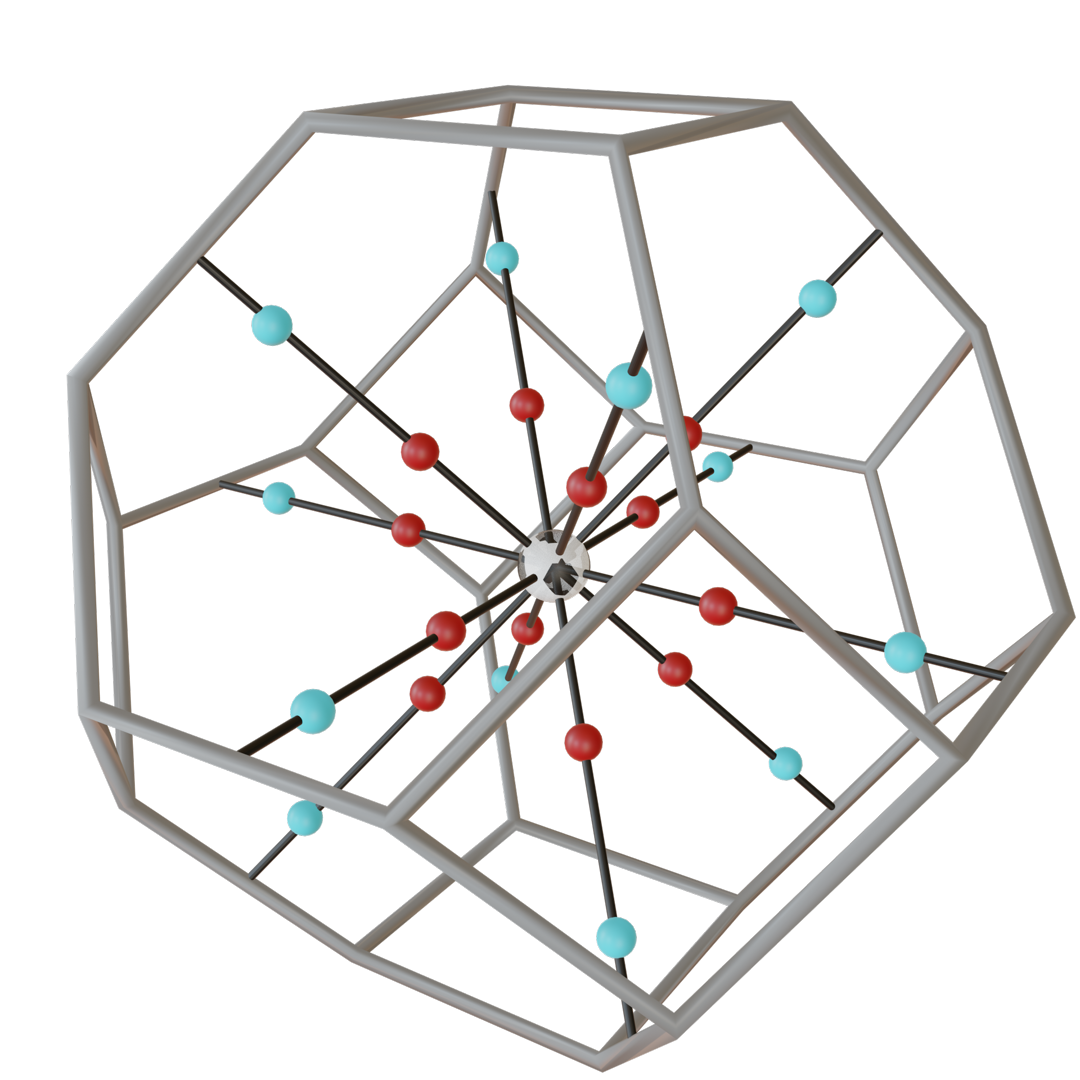} \\
		\caption{\footnotesize Equilibria of the electrostatic potential generated by unit point charges at the vertices of the truncated octahedron.
			In total there are $36$ \emph{light blue} $2$-saddles, $18$ \emph{dark red} $1$-saddles, and the degenerate equilibrium at the center.
			For better visualization, we split these equilibria into two groups, with one $1$-saddle and four $2$-saddles near each of the six squares displayed in the \emph{left panel}, and one $1$-saddle as well as one $2$-saddle for each of the twelve edges shared by two hexagons in the \emph{right panel}.}
		\label{fig:BCC-counterexample}
	\end{figure}
	
	\smallskip
	Making use of the obvious symmetries, we may associate equilibria of the electrostatic potential with the facets and edges of the polytope.
	For example, the line that passes through the centers of two opposite squares is the intersection of four planes of symmetry, and it passes through three equilibria: the center and a $1$-saddle near the square on either side.
	A more interesting example is the line that passes through the midpoints of two opposite edges shared by two hexagons each, which is the intersection of two planes of symmetry; see Figure~\ref{fig:BCC-counterexample}.
	We observe that such a line passes through five equilibria: the center and two saddles on each side of the center.
	It is the only example we have so far, in which an edge of the polytope seems to be associated with more than one equilibrium.
	Furthermore, equlibria located along lines obtained as the common fixed locus of at least two reflections can be deduced to exist using the method implemented in Appendix~\ref{app:A} for the cube.

	\section{Distance Functions}
	\label{sec:4}
	
	Conjectures~1.8 and 1.9 in \cite{GNS07} motivate us to take a closer look at the Voronoi tessellations defined by the Euclidean distance function and its weighted version.
	The former relates to the case in which all points have equal charges, while the latter models the situation is which all charges are positive but not necessarily equal to each other.
	
	\subsection{Euclidean Distance}
	\label{sec:4.1}
	
	If the number of equilibria were monotonically non-decreasing with growing $p$, then the number of such points of $V$ could be bounded from above by the number of such points of $E$.
	To study the latter, we introduce the \emph{Voronoi domain} of $A_i$ as the points $x \in \R^3$ that are at least as close to $A_i$ as to any other $A_j$.
	This domain is a closed convex polyhedron, and the \emph{Voronoi tessellation} is the collection of such domains, one for each point \cite{Vor070809}.
	The intersection of any collection of Voronoi domains is either empty or a shared face, which we refer to as a \emph{Voronoi cell}.
	\Skip{Any two Voronoi domains have disjoint interiors, and if they intersect, they do so along a common face.
		It follows that the interiors of the Voronoi domains and all their faces form a partition of $\R^3$ into open convex polyhedra, which we refer to as \emph{Voronoi cells}.}
	The dual of the Voronoi tessellation is the \emph{Delaunay mosaic} or \emph{Delaunay triangulation} of the points \cite{Del34}.
	It consists of convex polytopes that are convex hulls of subsets of the points.
	Specifically, the \emph{Delaunay cell} that corresponds to a Voronoi cell is the convex hull of the points that generate the Voronoi domains sharing the Voronoi cell, and the Delaunay mosaic is the collection of Delaunay cells.
	Generically, the Delaunay mosaic is a simplicial complex, and in general it is a polyhedral complex.
	
	\smallskip
	Following \cite{GNS07}, we call a Voronoi cell \emph{effective} if its interior has a non-empty intersection with the interior of the corresponding Delaunay cell.
	Corresponding cells have complementary dimensions and lie in orthogonal affine subspaces, which implies that the intersection is either empty or a point.
	Effective cells are interesting because they are in bijection with the equilibria of $E$.
	\begin{figure}[hbt]
		\vspace{-0.1in}
		\centering
		\includegraphics[width=.42\textwidth]{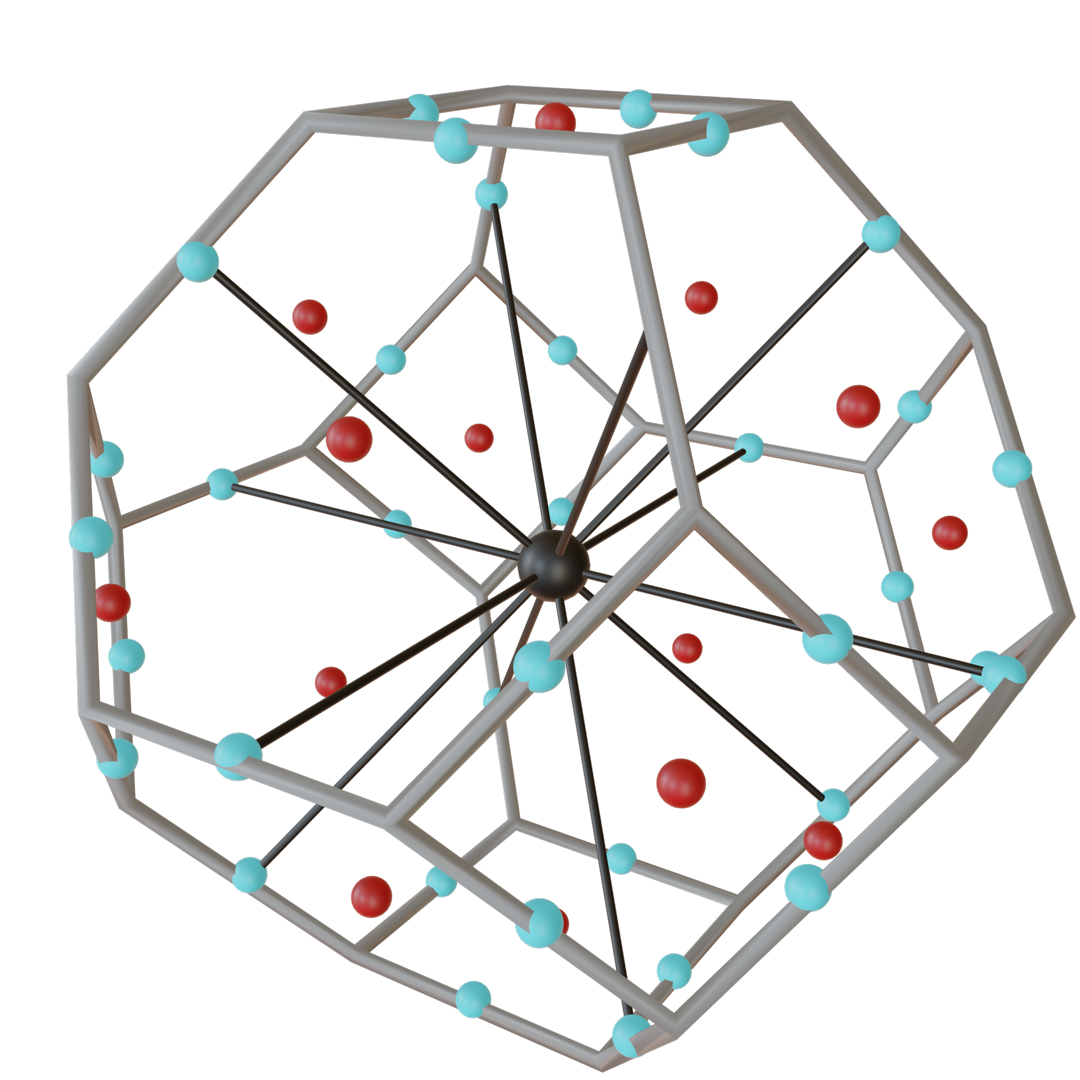} \\
		\vspace{-0.0in}
		\caption{\footnotesize The equilibria of $V_2$ generated by unit point charges at the vertices of the truncated octahedron.
			Compared with $V_1$, we note a drastically reduced number of $1$-saddles and a minimum at the origin; see Figure~\ref{fig:BCC-counterexample} where we used two copies of the solid to show all equilibria.
			While $p = 2$ is still not large, the equilibria are already close to the barycenters of the facets and edges of the solid.}
		\label{fig:BCC-equilibria}
	\end{figure}
	Indeed, with increasing $p$, the equilibria tend toward the intersections of corresponding Voronoi and Delaunay cells; see Figure~\ref{fig:BCC-equilibria}, which illustrates that already for $p=2$ the equilibria are barely distinguishable from these intersection points.
	The Upperbound Theorem for convex polytopes \cite{McM71} implies that the Voronoi tessellation of $n$ points in $\R^3$ has at most $O(n^2)$ cells, and thus at most that many effective cells.
	Recently, \cite[Theorem~3.1]{EdPa23} proved that this bound is asymptotically tight:
	\begin{theorem}
		\label{thm:lower_bound}
		For every $k \geq 2$, there exist $2k+2$ points in $\R^3$ such that the Voronoi tessellation has $(k+1)^2 -1$ effective $1$-dimensional and $k^2$ effective $2$-dimensional cells.
	\end{theorem}
	Does this lower bound translate to any meaningful lower bound for the electrostatic potential defined by the same number of points with equal charges?
	Perhaps not since the construction in \cite{EdPa23} is delicate and the authors of this article have not been able to get a quadratic number of equilibria of $V$ for similarly placed points.

	\subsection{Weighted Euclidean Distance}
	\label{sec:4.2}
	
	If we allow points with different charges, then $E$ is the multiplicatively weighted Euclidean distance function; see \eqref{eqn:Edistance}.
	We can still define Voronoi domains, tessellation, and cells, but there are significant differences:
	\begin{itemize}
		\item a domain is no longer the common intersection of closed half-spaces but of $n-1$ closed balls and closed complements of balls, and thus not necessarily connected;
		\item already in $\R^2$, the tessellation may consists of $\Theta (n^2)$ vertices, (circular) edges, and (connected) regions \cite{AuEd84};
		\item there is no established dual structure, like the Delaunay mosaic in the unweighted case.
	\end{itemize}
	In the plane, the maximum numbers of cells in the Voronoi tessellation differ even asymptotically between the unweighted and the weighted cases: $\Theta (n)$ versus $\Theta (n^2)$.
	In $\R^3$ on the other hand, the maximum number of cells is $\Theta (n^2)$ in both cases.
	There is at most one equilibrium of $E$ per Voronoi cell, so we have at most $\Theta (n^2)$ such points.
	However, the difference in the maximum number of cells in $\R^2$ begs the question whether the maximum number of equilibria differ in the same way.
	We think not.
	\begin{conjecture}
		\label{conj:2D_weighted_Euclidean_distance}
		The number of equilibria of the weighted Euclidean distance function defined by $n$ points with positive real weights in $\R^2$ is at most some constant times $n$.
	\end{conjecture}
	The number of equilibria in the plane is relevant because they can be turned into higher index equilibria in three dimensions.

	\subsection{Slices}
	\label{sec:4.3}
	
	Consider the Voronoi tessellation that represents the weighted or unweighted Euclidean distance function defined by $n$ weighted or unweighted points in $\R^3$.
	We call the restriction to a plane or a line a \emph{slice}.
	These restrictions look a lot like the lower-dimensional tessellations, but they are more general and can be defined by adding a weight to the respective squared distance to the $i$-th point:
	\begin{align}
		\pi_i (x) &= \dist{x}{A_i}^2 + w_i ; 
		\label{eqn:power} \\
		\varphi_i (x) &= \left( \frac{\dist{x}{A_i}}{\zeta_i} \right)^2 + w_i ,
		\label{eqn:wwdistance}
	\end{align}
	in which $w_i$ is the squared distance of $A_i$ from the plane or line.
	The weighted squared distance in \eqref{eqn:power} is known as \emph{power} or \emph{Laguerre distance}, and while it is a special case of \eqref{eqn:wwdistance}, it received substantially more attention in the mathematical literature.
	In both cases, the $1$-dimensional tessellation has at most $O(n)$ cells and thus at most that many equilibria.\footnote{Indeed, for the weighted squared distance in \eqref{eqn:power}, each point generates at most one interval along the line.  This is not true for the weighted squared distance in \eqref{eqn:wwdistance}, but it is not difficult to see that the point with minimum charge generates a single convex cell in three dimensions and therefore at most one interval along the line.
		The linear bound now follows by induction.}
	This is also true for the $2$-dimensional tessellation defined by the weighted squared distance in \eqref{eqn:power}, but see Conjecture~\ref{conj:2D_weighted_Euclidean_distance} to the more general weighted squared distance in \eqref{eqn:wwdistance}.
	By reducing to the same denominator and appealing to the fundamental theorem of algebra, one obtains that, for even $p > 0$ and $V_p$ defined by $n$ unit point charges, its restriction to a straight line in $\R^3$ has at most $p (n-1) + 2n-1$ equilibria.
	While we have no proof, we venture that this bound can be strengthened to at most $O(n)$ equilibria for any $p > 0$, and extended to $2$-dimensional slices.
	To focus, we formulate a more restricted version:
	\begin{conjecture}
		\label{conj:slices}
		Let $V \colon \R^3 \setminus \{A_1,A_2,\ldots,A_n\} \to \R$ be the electrostatic potential as in \eqref{eqn:V}.
		Then the restriction of $V$ to any straight line or flat plane in $\R^3$ has at most $O(n)$ equilibria.
	\end{conjecture}

	\section{Discussion}
	\label{sec:7}
	
	Inspired by the article of Gabrielov, Novikov and Shapiro \cite{GNS07}, we studied the connection between the electrostatic potential and the (weighted) Euclidean distance function defined by a finite collection of point charges.
	We focus on the $3$-dimensional case and on positive charges, which may or may not all be the same.
	We have three main results, the second of which is a counterexample to bounding the number of equilibria of the electrostatic potential by those of the Euclidean distance function.
	We therefore wish to suggest a more nuanced formulation of the question inspired by a result about the heat-flow in \cite{ChEd11}:
	\begin{quote}
		Is the $L_1$-norm of the persistence diagram of $V_p$ always smaller than that of $V_q$ for any $1 \leq p \leq q$?
		In other words, does the combined ``strength'' of the equilibria monotonically decrease with shrinking parameter, $p$?
	\end{quote}
	Such decay in strength does not contradict the increase in number of equilibria, which can indeed happen, as shown by our counterexample.
	The main question of determining the maximum number of equilibria of an electrostatic potential defined by $n$ point charges in $\R^3$ remains open.


	\Skip{
		\bibstyle{plainurl}
		\bibliography{references}
	} 
	
	\clearpage
	\appendix

	\section{Analytic Proof of Existence of Equilibria}
	\label{app:A}
	
	This appendix illustrates how we can prove the existence of claimed equilibria.
	We shall do this for the cube; see the second row of Table~\ref{tbl:solids}, but the argument generalizes to the other Platonic solid with more electrostatic points beside the center, i.e. the dodecahedron.
	
	\begin{proposition}
		\label{prop:equilibria_of_cube}
		The electrostatic potential defined by the (eight) unit point charges located at the vertices of a cube in $\R^3$ has twelve equilibria beside the one at the cube's center. These occur along the segments connecting the cube's center to the midpoints of its edges.
	\end{proposition}
	\begin{proof}
		Assume without loss of generality, that the vertices of the cube are located at the points $(\pm 1, \pm 1, \pm 1)$.
		The cube is invariant by a discrete group of symmetries, $G$, which contains, in particular, the reflections on the planes $P_1 = \{(x_1,x_2,x_3) \in \R^3 \ | \ x_2=0 \}$ and $P_2 = \{ (x_1,x_2,x_3) \in \R^3 \ | \  x_1+x_3=0 \}$. 
		Denote these reflections $\sigma_1, \sigma_2 \colon \R^3 \to \R^3$, respectively. 
		Letting $v_1 = (0,1,0)$ and $v_2 = (1,0,1)$ be vectors orthogonal to these planes, the reflections change their signs:
		$\sigma_1 (v_1) = -v_1$ and $\sigma_2 (v_2) = -v_2$.
		Furthermore, the electrostatic potential $V$ generated by unit charges at the vertices of this cube is invariant by these symmetries, i.e.\ $V \circ \sigma_1 = V \circ \sigma_2 = V$.
		Hence, we find from the symmetry that, along $P_i$,
		\begin{align}
			\langle \nabla V , v_i \rangle = \langle \sigma_i (\nabla V) , \sigma_i (v_i) \rangle =  - \langle \nabla V , v_i \rangle ,
		\end{align}
		and therefore $\langle \nabla V , v_i \rangle=0$, for $i = 1,2$.
		In particular, along 
		\begin{align}
			P_1 \cap P_2 & = \{ (x_1,x_2,x_3) \in \R^3 \ | \ x_2=0 , \ x_1+x_3=0 \} \\
			& = \{ (x,0,-x) \in \R^3 \ | x \in \R \}
		\end{align}
		we have $\langle \nabla V, v_1 \rangle = 0$ and $\langle \nabla V , v_2 \rangle = 0$, which we write more explicitly as
		\begin{align}
			\label{eq:Derivatives symmetries}
			\frac{\partial V}{\partial x_2}(x,0,-x) = \frac{\partial V}{\partial x_1}(x,0,-x) + \frac{\partial V}{\partial x_3}(x,0,-x) = 0. 
		\end{align}
		We shall now prove the existence of an equilibrium of $V$ by finding one in $P_1 \cap P_2$.
		Due to \eqref{eq:Derivatives symmetries}, we need only prove that there is an $x \in \R$ such that the function $f(x)=V(x,0,-x)$ has a critical point, i.e.\ we must find $x$ such that $f'(x) = 0$; that is:
		\begin{align}
			\frac{\partial V}{\partial x_1}(x,0,-x) - \frac{\partial V}{\partial x_3} (x,0,-x) &= 0.
		\end{align}
		The function $f'(x)$ is odd and varies continuously with $x$.
		We can explicitly compute it and therefore evaluate it at any point; see Figure~\ref{fig:PotentialLine}. 
		\begin{figure}[hbt]
			\centering \vspace{0.0in}
			\includegraphics[width=0.35\textwidth]{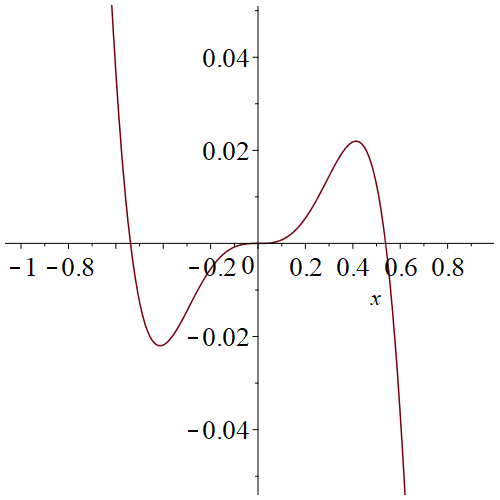}
			\vspace{0.0in}
			\caption{\footnotesize Plot of $f'(x)$ whose zeroes correspond to electrostatic points lying in the line $P_1 \cap P_2$.}
			\label{fig:PotentialLine}
		\end{figure}
		Evaluating the explicit formula for $f'(x)$ we find that
		\begin{align}
			f'(-1) & = \tfrac{8}{27} + \tfrac{8 \sqrt{5}}{25} > 0,\\
			f'(-1/2) & = - \tfrac{4}{63} \left[ \sqrt{7} \left( \sqrt{3} - \tfrac{27 \sqrt{11}}{121} \right) - \tfrac{18}{7} \right] < 0, \\
			f'(1/2) & =  \tfrac{4}{63} \left[ \sqrt{7} \left( \sqrt{3} - \tfrac{27 \sqrt{11}}{121} \right) - \tfrac{18}{7} \right] >0, \\
			f'(1) & = \tfrac{8}{27} + \tfrac{8 \sqrt{5}}{25} > 0.
		\end{align}
		We therefore conclude from the intermediate value theorem that there are two extra equilibria of $V$ (zeroes of $f'(x)$) located at points $(x_*,0,-x_*)$ and $(-x_*,0,x_*)$ for $x_* \in (1/2, 1)$. Notice, in particular, that such points lie along the straight lines passing through the cube's center and the midpoints in its edges.
		
		These two points are in the same orbit of the group of symmetries of the cube and this orbit has cardinality $12$ thus resulting in such a number of extra equilibria.
	\end{proof}

	\section{Archimedean and Catalan Solids}
	\label{app:B}
	
	In this appendix, we list the numbers of equilibria of the electrostatic potentials generated by placing unit point charges at the vertices of Archimedean and Catalan solids.
	We recall that an \emph{Archimedean solid} is a convex polytope that is neither a Platonic solid nor a prism or anti-prism such that
	\begin{itemize}
		\item all facets are regular polygons;
		\item all vertices have the same local shape.
	\end{itemize}
	Conventially, there are $13$ Archimedean solids, but depending on the interpretation of the second condition, there may be a $14$-th.
	In particular, we may or may not require vertex-transitivity (for any two vertices, there is a symmetry that maps one vertex into the other), which holds for the $13$ Archimedean solids but not for the Elongated Square Gyrobicupola.
	The latter is similar to the Rhombicuboctahedron but has one of the square cuppola rotated by $45^\circ$ relative to the opposite cupola; see Figure~\ref{fig:rhombicuboctahedron-equilibria}.
	Comparing rows $5$ and $14$ in Table~\ref{tbl:Archimedean-solids}, we see that the Elongated Square Gyrobicupo generates fewer equilibria than the Rhombicuboctahedron.
	\begin{table}[hbt]
		\centering \footnotesize
		\begin{tabular}{l||c|ccc}
			\multicolumn{1}{c||}{solid} & $f$-vector     & \multicolumn{1}{c}{2-saddles} & \multicolumn{1}{c}{1-saddles} & \multicolumn{1}{c}{center}  \\ \hline \hline
			Truncated Tetrahedon            & (12,18,8,1)   &  18 &  4 &  -3   \\
			Cuboctahedron                   & (12,24,14,1)  &  24 &  8 &  -5   \\
			Truncated Cube                  & (24,36,14,1)  &  36 &  8 &  -5   \\
			Truncated Octahedron            & (24,36,14,1)  &  36 & 18 &   5   \\
			Rhombicuboctahedron             & (24,48,26,1)  &  36 &  8 &  -5   \\
			Truncated Cuboctahedron         & (48,72,26,1)  &  72 & 20 &  -5   \\
			Snub Cube                       & (24,60,38,1)  &  36 &  8 &  -5   \\
			Icosidodecahedron               & (30,60,32,1)  &  60 & 20 &  -11  \\
			Truncated Dodecahedron          & (60,90,32,1)  &  90 & 20 &  -11  \\
			Truncated Icosahedron           & (60,90,32,1)  &  90 & 12 &  -19  \\
			Rhombicosidodecahedron          & (60,120,62,1) &  90 & 20 &  -11  \\
			Truncated Icosidodecahedron     & (120,180,62,1)& 180 & 50 &  -11  \\
			Snub Dodecahedron               & (60,150,92,1) &  90 & 20 &  -11  \\
			Elongated Square Gyrobicupola   & (24,48,26,1)  &  32 &  8 &  -1
		\end{tabular}
		\vspace{0.1in}
		\caption{\footnotesize The equilibria of the electrostatic potential of point charges at the vertices of the thirteen Archimedean solids and the Elongated Square Gyrobicupola.
			The $f$-vector gives the number of vertices, edges, facets, and the solid itself, in this sequence.
			Except possibly in the last case, the center is a degenerate equilibrium, and in each case, the number of $1$-saddles is non-zero; compare with Table~\ref{tbl:Catalan-solids}.}
		\label{tbl:Archimedean-solids}
	\end{table}
	
	\Skip{\smallskip
		HI CHRIS.  THIS IS GOOD TEXT BUT A BIT TOO ELABORATE AND NOT INTEGRATED.  PERHAPS SOME OF IT CAN BE ADDED AT THE RIGHT PLACE.
		The examples in Figure~\ref{fig:rhombicuboctahedron-equilibria} are the so-called Rhombicuboctahedron and the elongated square gyrobicupola, whose boundary consists of $16$ squares and $8$ triangles.  The elongated square gyrobicupola has half the vertices rotated by $90^\circ$.
		While there are $24$ facets, we observe only $8$ $1$-saddles, presumably one for each triangle in both cases.
		Similarly, there are $48$ edges, but we observe only $36$ $2$-saddles in the rhombicuboctahedron and $32$ $2$-saddles in the elongated square gyrobicupola, presumably one for each triangle edge and an additional one for each square neighbouring two triangles where super level set components corresponding to the triangles merge.  The later case occuring more frequently in the rhombicuboctahedron.
	} 
	\begin{figure}[hbt]
		\vspace{-0.1in}
		\centering
		\includegraphics[width=0.38\textwidth]{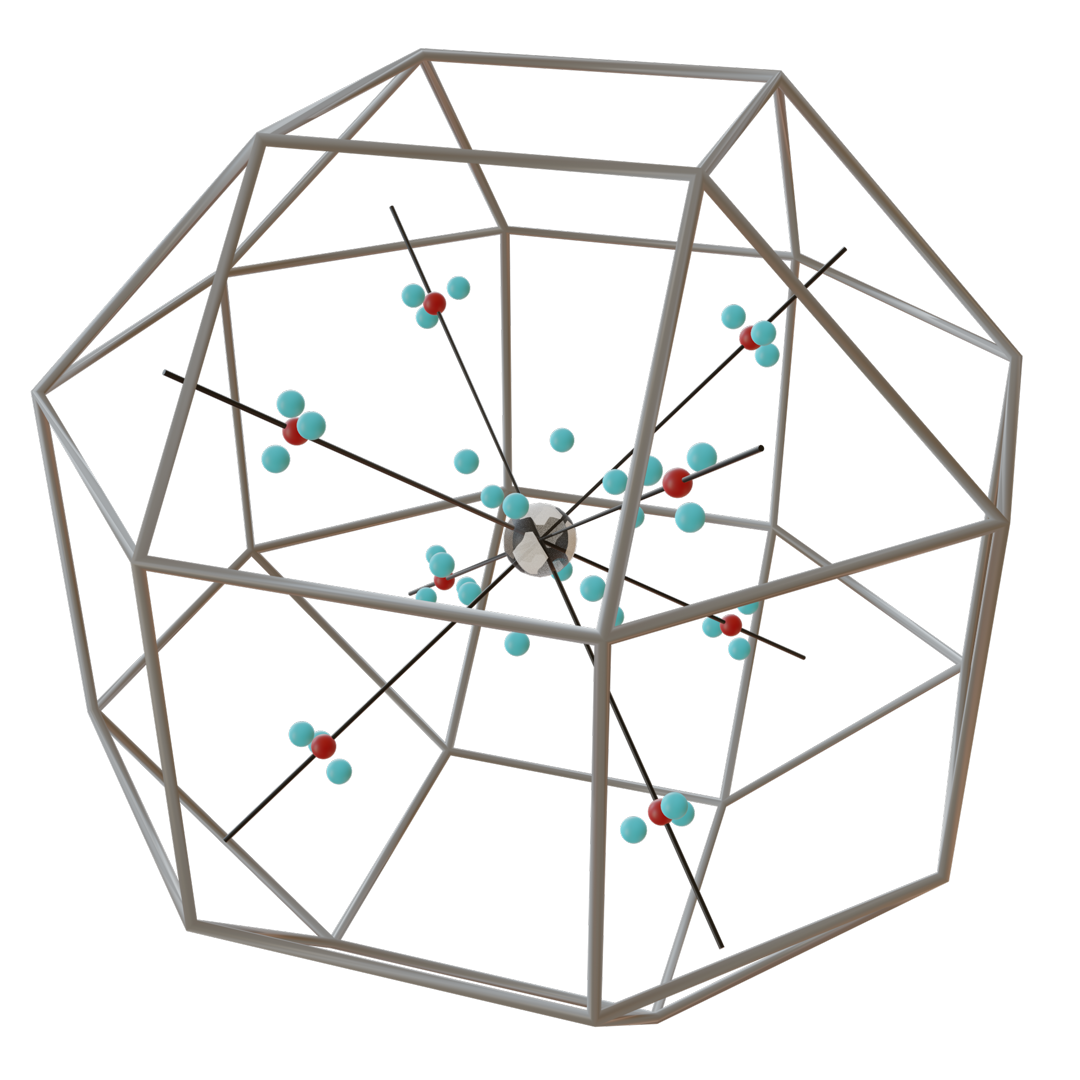} \hspace{0.2in}
		\includegraphics[width=0.38\textwidth]{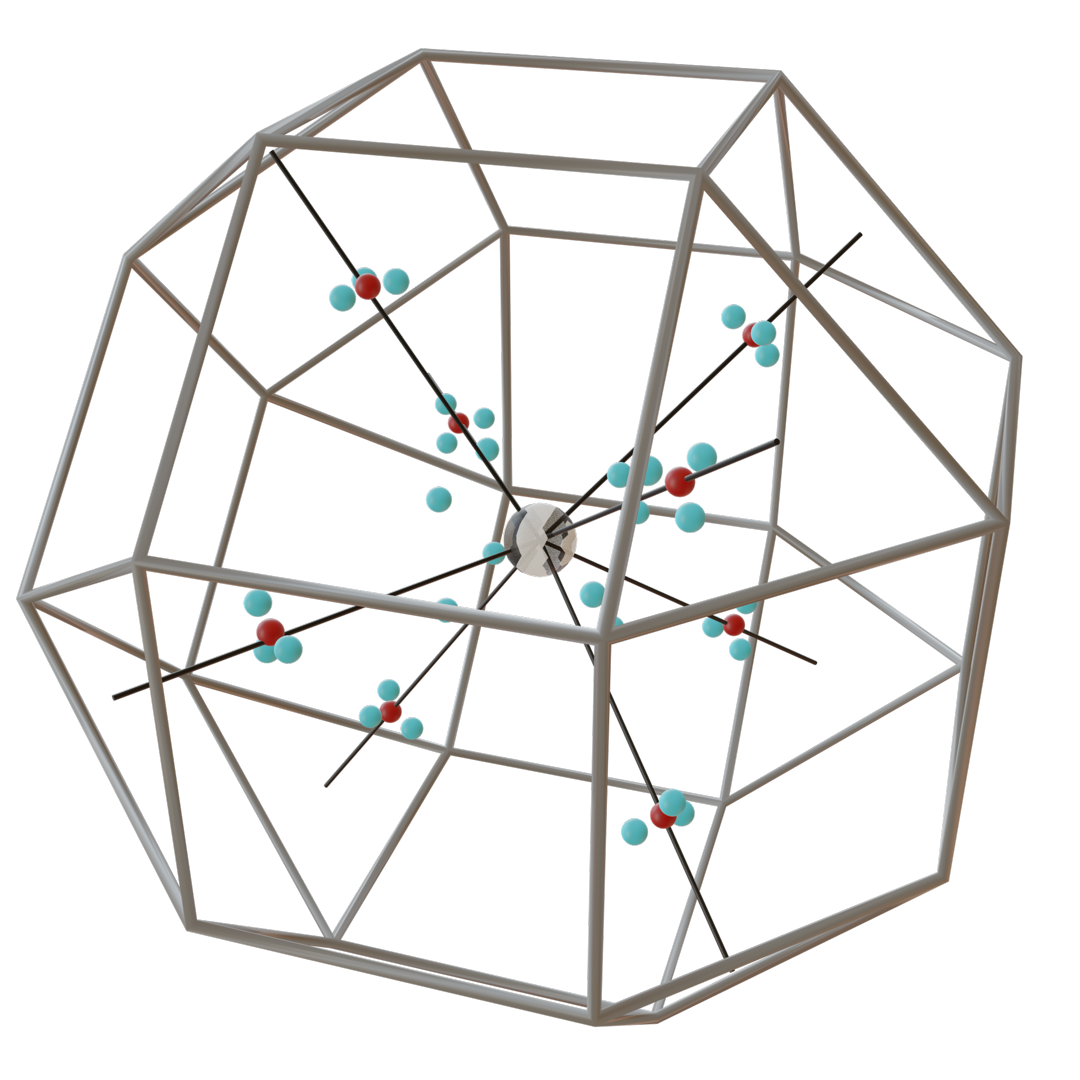} \\
		\vspace{-0.1in}
		\caption{\footnotesize \emph{Left:} the Rhombicuboctahedron and the equilibria of the electrostatic potential generated by unit point charges at its vertices.
			There are $36$ \emph{light blue} $2$-saddles, $8$ \emph{dark red} $1$-saddles, and a degenerate equilibrium at the center.
			\emph{Right:} the Elongated Square Gyrobicupola---which differs from the Rhombicuboctahedron by the left cupola being rotated by $45^\circ$--- and the equilibria of the electrostatic potential generated by unit point charges at its vertices.
			There are fewer equilibria: $32$ \emph{light blue} $2$-saddles, $8$ \emph{dark red} $1$-saddles, and a degenerate equilibrium at the center.  
		}
		\label{fig:rhombicuboctahedron-equilibria}
	\end{figure}
	
	The \emph{Catalan solids} are dual to the Archimedean solids, so the vertices have regular local shapes, and the facets are congruent to each other.
	There are again $13$ classical examples, and the Pseudo Deltoidal Icositegrahedron as a $14$-th solid that is not facet-transitive.
	Comparing rows $5$ and $14$ in Table~\ref{tbl:Catalan-solids}, we see that the Pseudo Deltoidal Icositegrahedron generates again fewer equilibria than the similar Deltoidal Icositetrahedron.
	\begin{table}[hbt]
		\centering \footnotesize
		\begin{tabular}{l||c|ccc}
			\multicolumn{1}{c||}{solid} & $f$-vector     & \multicolumn{1}{c}{2-saddles} & \multicolumn{1}{c}{1-saddles} & \multicolumn{1}{c}{center}  \\ \hline \hline
			Triakis Tetrahedron                 & (8,18,12,1)       &  10 &  0 &  -3   \\
			Rhombic Dodecahedron                & (14,24,12,1)      &  18 &  0 &  -5   \\
			Triakis Octahedron                  & (14,36,24,1)      &  18 &  0 &  -5   \\
			Tetrakis Hexehedron                 & (14,36,24,1)      &  20 &  0 &  -7   \\
			Deltoidal Icositetrahedron          & (26,48,24,1)      &  30 &  0 &  -5   \\
			Disdyakis Dodecahedron              & (26,72,48,1)      &  30 &  0 &  -5   \\
			Pentagonal Icositetrahedron         & (38,60,24,1)      &  48 &  0 &  -11  \\
			Rhombic Triacontahedron             & (32,60,30,1)      &  42 &  0 &  -11  \\
			Triakis Icosahedron                 & (32,90,60,1)      &  42 &  0 &  -11  \\
			Pentakis Dodecahedron               & (32,90,60,1)      &  50 &  0 &  -19  \\
			Deltoidal Hexecontahedron           & (62,120,60,1)     &  72 &  0 &  -11  \\
			Disdyakis Triacontahedron           & (62,180,120,1)    &  72 &  0 &  -11  \\
			Pentagonal Hexecontahedron          & (92,150,60,1)     & 102 &  0 &  -11  \\
			Pseudo Deltoidal Icositetrahedron   & (26,48,24,1)      &  26 &  0 &  -1   
		\end{tabular}
		\vspace{0.1in}
		\caption{\footnotesize  The equilibria of the electrostatic potential of point charges at the vertices of the thirteen Catalan solids and the Pseudo Deltoidal Icositetrahedron.
			The $f$-vector gives the number of vertices, edges, facets, and the solid itself, in this sequence.
			Except possibly in the last case, the center is a degenerate equilibrium, and in each case, the number of $1$-saddles vanishes; compare with Table~\ref{tbl:Archimedean-solids}.}
		\label{tbl:Catalan-solids}
	\end{table}
	
	\smallskip
	Finally, we note that all Archimedean solids define electrostatic potentials with a positive number of (non-degenerate) $1$-saddles, while all Catalan solids generate no such $1$-saddles; compare Tables~\ref{tbl:Archimedean-solids} and \ref{tbl:Catalan-solids}.
	Curiously, also the Platonic solids generate no $1$-saddles; see Table~\ref{tbl:solids}.
	At this time, the authors of this paper have no explanation for this curious observation.

	\section{Prisms and Anti-prisms}
	\label{app:C}
	
	In this appendix, we provide details about the equilibria of the electrostatic potential for unit point charges placed at the vertices of regular prisms and anti-prisms.
	After introducing the building blocks of the analysis, we focus on the square and triangular anti-prisms.
	
	\subsection{Regular polygons}
	\label{app:C.1}
	
	Let $V_0$ be the electric potential generated by unit point charges located at the vertices of a regular $N$-gon centered at the origin, $0 \in \R^3$, and placed on the horizontal plane spanned by the first two coordinate vectors. 
	Denote by $R$ the radius of the circle that passes through its vertices, and by $\alpha = \frac{2\pi}{N}$ the smallest angle that is a rotation symmetry of the $N$-gon.
	We start by proving that the partial derivatives are symmetric with respect to the horizontal plane, and they all vanish at the origin.
	\begin{lemma}
		\label{lem:derivatives_for_prism}
		Let $V_0 \colon \R^3 \to \R$ be the electrostatic potential of unit point charges at the vertices of a regular $N$-gon, as described, and write $x=(x_1,x_2,x_3)$ for a point with its Cartesian coordinates in $\R^3$.
		Then
		\begin{align}
			\frac{\partial V_0}{\partial x_3} (x_1,x_2,x_3) = - \frac{\partial V_0}{\partial x_3} (x_1,x_2,-x_3)
			\mbox{\rm ~~and~~}
			\frac{\partial V_0}{\partial x_1} (0,0,x_3) &= \frac{\partial V_0}{\partial x_2} (0,0,x_3) = 0.
		\end{align}
	\end{lemma}
	\begin{proof}
		The configuration is invariant by reflection through the $x_3=0$ plane, which contains the $N$-gon.
		Hence, $V_0(x_1,x_2,x_3) = V_0(x_1,x_2,-x_3)$ and differentiating with respect to $x_3$ gives the first relation.
		To see the second relation, note that the invariance of the $N$-gon under rotations of the horizontal plane by an angle $\alpha$ implies
		\begin{align} 
			V_0(x_1,x_2,x_3) &= V_0(x_1 \cos{\alpha} + x_2 \sin{\alpha}, -x_1 \sin{\alpha} + x_2 \cos{\alpha} , x_3),
		\end{align}
		for all $x=(x_1,x_2,x_3) \in \R^3$. 
		Differentiating with respect to $x_1$ and $x_2$, respectively, and evaluating at $x_1 = x_2 = 0$ yields
		\begin{align}
			\frac{\partial V_0}{\partial x_1} (0,0,x_3) &= \cos{\alpha} \frac{\partial V_0}{\partial x_1} (0,0,x_3) - \sin{\alpha} \frac{\partial V_0}{\partial x_2} (0,0,x_3) ; \\
			\frac{\partial V_0}{\partial x_2} (0,0,x_3) &= \sin{\alpha} \frac{\partial V_0}{\partial x_1} (0,0,x_3) + \cos{\alpha} \frac{\partial V_0}{\partial x_2} (0,0,x_3) .
		\end{align}
		This gives a system of linear equations for $\frac{\partial V_0}{\partial x_1} (0,0,x_3)$ and $\frac{\partial V_0}{\partial x_1} (0,0,x_3)$ whose only solution is when both vanish, as claimed.
	\end{proof}
	
	It will also be useful to have the second partial derivative at the origin in vertical direction.
	\begin{lemma}
		\label{lem:eigenvalue_for_prism}
		For $V_0 \colon \R^3 \to \R$ as before and all $x_3 \in \R$, we have 
		\begin{align}
			\frac{\partial^2 V_0}{ (\partial x_3)^2} (0,0, x_3 ) &= - \frac{N}{(R^2+x_3^2)^{5/2}}\left( R^2 - 2 x_3^2 \right).
			\label{eqn:E2}
		\end{align}
	\end{lemma}
	\begin{proof}
		Denote the vertices of the $N$-gon by $A_i = (a_{i,1},a_{i,2},0)$, for $i=1,2, \ldots , N$, and write $r_i$ for the distance of $A_i$ from a point $x = (x_1, x_2, x_3)$.
		Then the first and second partial derivatives along the vertical direction are
		\begin{align}
			\frac{\partial V_0}{ \partial x_3} 
			&= -\sum_{i=1}^N \frac{1}{r_i^{2}} \frac{\partial r_i}{\partial x_3} 
			= -\sum_{i=1}^N \frac{x_3}{r_i^{3}} ; \\
			\frac{\partial^2 V_0}{(\partial x_3)^2} 
			&= -\sum_{i=1}^N \left( \frac{1}{r_i^{3}} - 3\frac{x_3}{r_i^{4}} \frac{\partial r_i}{\partial x_3} \right)
			= -\sum_{i=1}^N \left(  \frac{1}    {r_i^{3}} - 3 \frac{x_3^2}{r_i^{5}} \right)
			= -\sum_{i=1}^N \frac{1}{r_i^{5}}\left( r_i^2 - 3 x_3^2 \right).
		\end{align}
		Evaluating the second partial derivative at $(0,0,x_3)$, we have $r_i^2=R^2+x_3^2$ independent of $i$, which implies \eqref{eqn:E2}.
	\end{proof}
	
	\subsection{Rotated prisms}
	\label{app:C.2}
	
	We get a prism by connecting a regular $N$-gon with a copy of itself translated in an orthogonal direction.
	If we first rotate the copy by an angle $\beta$---around its center and within its plane---we call the convex hull of the two $N$-gons a \emph{$\beta$-rotated prism}.
	For $\beta = \frac{\alpha}{2}$, this is an anti-prism.
	Its \emph{axis} is the line that passes through the centers of the two $N$-gons, its \emph{radius} is the distance of the vertices from this line, and its \emph{height} is the distance between the planes that contain the two $N$-gons.
	We now consider the electrostatic potential obtained by placing unit point charges at the vertices of such a $\beta$-rotated prism.
	Supposing the $N$-gons lie in horizontal planes at distance $\frac{h}{2}$ above and below the origin, the electrostatic potential is
	\begin{align}
		V(x_1,x_2,x_3) &= V_0(x_1,x_2,x_3+\tfrac{h}{2}) \nonumber \\
		&+ V_0(x_1 \cos{\beta} + x_2 \sin{\beta}, -x_1 \sin{\beta} + x_2 \cos{\beta} ,x_3-\tfrac{h}{2}) .
		\label{eqn:beta-potential}
	\end{align}
	Not surprisingly, the type of the equilibrium at the origin depends on the height, but with the exception of a particular height, it does not depend on the angle of the rotation.
	\begin{theorem}\label{thm:center_of_prism}
		Let $\beta \in \R / 2 \pi \Z$, $R>0$, and $h>0$. 
		Consider the electrostatic potential generated by unit point charges at the vertices of a $\beta$-rotated prism with radius $R$ and height $h$. 
		Then, its center is a non-degenerate equilibrium iff $h \neq \sqrt{2}R$, in which case the center is a $1$-saddle if $h < \sqrt{2} R$ and a $2$-saddle if $h > \sqrt{2} R$.
		Moreover, in the latter case there are two additional $1$-saddles along the axis of the $\beta$-rotated prism.
	\end{theorem}
	\begin{proof}
		To use the form of the potential in \eqref{eqn:beta-potential}, we assume that the two $N$-gons of the $\beta$-rotated prism lie in horizontal planes and its center is the origin in $\R^3$.
		We start by confirming that the origin is always an equilibrium.
		Using Lemma~\ref{lem:derivatives_for_prism}, we obtain
		\begin{align}
			\frac{\partial V}{\partial x_i} (0,0,0) & = \frac{\partial V_0}{\partial x_i} (0,0,\tfrac{h}{2}) + \frac{\partial V_0}{\partial x_i} (0,0,- \tfrac{h}{2}) = 0 ,
		\end{align}
		which for $i = 1, 2$ is true because both partial derivatives vanish, and for $i = 3$ follows from the symmetry across the horizontal coordinate plane.
		Next we investigate the Hessian at the origin.
		By differentiating the identity $\frac{\partial V_0}{\partial x_3} (x_1,x_2,x_3) = - \frac{\partial V_0}{\partial x_3} (x_1,x_2,-x_3)$ with respect to $x_1$, $x_2$, and $x_3$, we obtain 
		\begin{align}
			\frac{\partial^2 V_0}{\partial x_i \partial x_3} (x_1,x_2,x_3) & = - \frac{\partial^2 V_0}{\partial x_i \partial x_3} (x_1,x_2,-x_3) ; \\
			\frac{\partial^2 V_0}{(\partial x_3)^2} (x_1,x_2,x_3) & = \frac{\partial^2 V_0}{ (\partial x_3)^2} (x_1,x_2,-x_3),
		\end{align}
		for $i = 1,2$.
		Hence, using the first of these identities, we have
		\begin{align}
			\frac{\partial^2 V}{\partial x_i \partial x_3}(0,0,0) = \frac{\partial^2 V_0}{\partial x_i \partial x_3} (0,0, \tfrac{h}{2}) + \frac{\partial^2 V_0}{\partial x_i \partial x_3} (0,0, -\tfrac{h}{2}) =0.
		\end{align}
		This shows that $\lambda_3 = \frac{\partial^2 V}{(\partial x_3)^2}(0,0,0)$ is an eigenvalue of the Hessian at the origin.
		Using the second identity above, we see that this eigenvalue satisfies
		\begin{align}
			\lambda_3 &= \frac{\partial^2 V}{(\partial x_3)^2}(0,0,0) = \frac{\partial^2 V_0}{\partial x_i \partial x_3} (0,0, \tfrac{h}{2}) + \frac{\partial^2 V_0}{\partial x_i \partial x_3} (0,0, -\tfrac{h}{2}) = 2 \frac{\partial^2 V_0}{(\partial x_3)^2} (0,0, \tfrac{h}{2}).
		\end{align}
		The rotational symmetry by $\beta \neq 0$ then implies that the remaining two eigenvalues must coincide: $\lambda_1 = \lambda_2 = \lambda$.
		Furthermore, since $V$ is harmonic, we have $\lambda = -\frac{\partial^2 V_0}{(\partial x_3)^2} (0,0, \frac{h}{2})$. 
		It remains to compute this quantity and show it vanishes precisely when $h = \sqrt{2} R$ as claimed in the statement. 
		This can be done by making use of Lemma \ref{lem:eigenvalue_for_prism} to obtain
		\begin{align}
			\frac{\partial^2 V}{(\partial x_3)^2}(0,0,0) = - \frac{2N}{(R^2+x_3^2)^{5/2}}\left( R^2 - \tfrac{h^2}{2} \right),
		\end{align}
		which vanishes iff $h=\sqrt{2}R$.
		Furthermore, the eigenvalues of the Hessian at the origin are   
		\begin{align}
			\lambda_3 & = - \frac{2N}{(R^2+x_3^2)^{5/2}}\left( R^2 - \tfrac{h^2}{2} \right) ; \\
			\lambda_1 = \lambda_2 & = \frac{N}{(R^2+x_3^2)^{5/2}}\left( R^2 - \tfrac{h^2}{2} \right) .
		\end{align}
		We therefore conclude that for $h<\sqrt{2}R$ we have $\lambda_1=\lambda_2>0$ and $\lambda_3<0$, so the origin is a $1$-saddle. On the other hand, if $h>\sqrt{2}R$ we have $\lambda_1=\lambda_2<0$ and $\lambda_3>0$, and so the origin is a $2$-saddle.
		However, in this case there are at least two additional equilibria along the axis of the $\beta$-rotated prism, which by construction is the $x_3$-axis.
		To prove their existences, we compute the derivatives of $V$ along this axis:
		\begin{align}
			\frac{\partial V}{\partial x_i} (0,0,x_3) & = \frac{\partial V_0}{\partial x_i} (0,0,x_3+\tfrac{h}{2}) + \frac{\partial V_0}{\partial x_i} (0,0,x_3- \tfrac{h}{2}) = 0 +0 =0 ,
		\end{align}
		for $i=1,2$. 
		Hence, we need only look for zeroes of $\frac{\partial V}{\partial x_3} (0,0,x_3)$, which we compute as
		\begin{align}
			\frac{\partial V}{\partial x_3} (0,0,x_3) &= \frac{\partial V_0}{\partial x_3} (0,0,x_3+\tfrac{h}{2}) + \frac{\partial V_0}{\partial x_3} (0,0,x_3- \tfrac{h}{2}) .
		\end{align}
		To compute the two terms on the right-hand side, we use Lemma \ref{lem:eigenvalue_for_prism} and find that $\frac{\partial V_0}{ \partial x_3} = -\sum\nolimits_{i=1}^N \frac{x_3}{r_i^{3}}$, where $r_i(x_1,x_2,x_3)^2=(x_1+a_{i1})^2+(x_2-a_{i2})^2+x_3^2$.
		From this it follows that
		\begin{align}
			\frac{\partial V}{\partial x_3} (0,0,x_3) &=  -\sum\nolimits_{i=1}^N \frac{x_3+\frac{h}{2}}{r_i^{3}(0,0,x_3+\frac{h}{2})}  -\sum\nolimits_{i=1}^N \frac{x_3- \frac{h}{2}}{r_i^{3}(0,0,x_3-\frac{h}{2})} \\
			&=  -N \frac{x_3+\frac{h}{2}}{\left(  R^2 + (x_3+\frac{h}{2})^2 \right)^{3/2}}  - N \frac{x_3- \frac{h}{2}}{\left(  R^2 + (x_3-\frac{h}{2})^2 \right)^{3/2}} .
		\end{align}
		Evaluating this at the top and bottom of the $\beta$-rotated prism, we obtain
		\begin{align}
			\frac{\partial V}{\partial x_3} (0,0,\tfrac{h}{2})=-\frac{h}{(R^2+h^2)^{3/2}} <0 \mbox{\rm ~~and~~} \frac{\partial V}{\partial x_3} (0,0,-\tfrac{h}{2})=\frac{h}{(R^2+h^2)^{3/2}}>0.
		\end{align}
		On the other hand, further differentiating with respect to $x_3$, as in the proof of Theorem \ref{thm:center_of_prism}, we find that
		\begin{align}
			\frac{\partial^2 V}{\partial x_3^2}(0,0,0) = - \frac{2N}{(R^2+x_3^2))^{5/2}}\left( R^2 - \tfrac{h^2}{2} \right),
		\end{align}
		which is positive when $h>\sqrt{2}R$. 
		Hence, there exists $\ee>0$, such that $\frac{\partial V}{\partial x_3} (0,0,t)$ is positive for $t \in (0,\ee)$ and negative for $t\in (-\ee,0)$. 
		Combining this with the previously computed values of $\frac{\partial V}{\partial x_3}$ at the centers of the two $N$-gons, the intermediate value theorem implies that $\frac{\partial V}{\partial x_3} $ has two additional zeroes along the $x_3$-axis.
		
		\smallskip
		We now characterize these equilibria by showing that they are $1$-saddles. 
		At both, $V(0,0,x_3)$ increases on the left and decreases on the right, so they are local maxima of the potential restricted to the $x_3$-axis. 
		Given that the potential is harmonic, at least one of the remaining eigenvalues of the Hessian must be positive, and by rotational symmetry so must the other. 
		As the Hessian has a single negative eigenvalue, both equilibria are $1$-saddles.
	\end{proof}
	
	\subsection{Existence results for equilibria in ati-prisms}
	\label{app:C.3}
	
	We shall consider an anti-prims whose top and bottom faces are $N$-gons and denote their minimum angle of symmetry by $\alpha=\frac{2\pi}{N}$. We shall let $R>0$ be the radius of these $N$-gons and $h>0$ the height of the anti-prism, then its relative height is defined by $\frac{h}{R}$.
	
	Identify $\mathbb{R}^3 \cong \mathbb{C} \times \mathbb{R}$, then for $\ell \in \{ 1, \ldots ,N \}$ we shall denote the vertices of the top $N$-gon as $A_\ell=(Re^{i (\ell \alpha + \frac{\alpha}{4})} , \frac{h}{2})$ and those at the bottom $N$-gon as $B_\ell=(Re^{-i (\ell \alpha + \frac{\alpha}{4})} , -\frac{h}{2})$. Then, the electric potential generated by placing unit point charges at the vertices of this anti-prism is
	$$V(x)=V_u(x)+V_d(x),$$
	where 
	$$V_u(x)=\sum_{\ell=1}^N\frac{1}{|x-A_\ell|} , \ \  V_d(x)=\sum_{\ell=1}^N\frac{1}{|x-B_\ell|}.$$ 
	We shall now state and prove the following result.
	
	\begin{theorem}
		Let $V$ be the electric potential generated by placing unit point charges at the vertices of an anti-prism as above and suppose its relative height satisfies $\frac{h}{R} < \sqrt{2}$. Then, there are electrostatic points located along the line segments connecting the center of the anti-prism to the midpoints of its lateral edges. 
	\end{theorem}
	
	\begin{proof}
		We shall consider the case of the line segment connecting the center to the midpoint of the edge connecting $A_N$ to $B_N$. The remaining cases can be handled similarly.
		
		The midpoint between $A_N$ and $B_N$ is $\frac{A_N+B_N}{2}=(R \cos(\frac{\alpha}{4}) , 0 , 0)$ and we can parametrize the segment that connects the origin to it using $t \in [0,1] \mapsto (t R \cos(\frac{\alpha}{4}) , 0 , 0)$. The derivatives of $V(x)$ in the normal directions $n_1=(0,1,0)$ and $n_2=(0,0,1)$ are given by
		\begin{align*}
			\langle \nabla V , n_1 \rangle & = \frac{\partial V_u}{\partial x_2} + \frac{\partial V_d}{\partial x_2} \\
			\langle \nabla V , n_2 \rangle & = \frac{\partial V_u}{\partial x_3} + \frac{\partial V_d}{\partial x_3}.
		\end{align*}
		Now, we compute each of the terms in the rights hand side in turn. We start with
		\begin{align*}
			\frac{\partial V_u}{\partial x_2} & = -\sum_{\ell=1}^N\frac{(x_2-R\sin(\ell \alpha + \frac{\alpha}{4}))}{|x-A_\ell|^3}  \\
			\frac{\partial V_d}{\partial x_2} & = -\sum_{\ell=1}^N\frac{(x_2+R\sin(\ell \alpha + \frac{\alpha}{4}))}{|x-B_\ell|^3} ,
		\end{align*}
		on the other hand, along the line connecting $0$ to $\frac{A_N+B_N}{2}$, we have $x_2=0=x_3$ and
		\begin{align*}
			|x-A_\ell| & = |(x_1-R\cos(\ell \alpha + \tfrac{\alpha}{4}) , R\sin(\ell \alpha + \tfrac{\alpha}{4}) , 0 )| \\
			& = |(x_1-R\cos( - \ell \alpha - \tfrac{\alpha}{4}) ,  R\sin( - \ell \alpha - \tfrac{\alpha}{4}) , 0 )| \\
			&  = |x-B_\ell| .
		\end{align*} 
		which shows that along such a line
		\begin{align*}
			\frac{\partial V_u}{\partial x_2} =  - \frac{\partial V_d}{\partial x_2} .
		\end{align*}
		On the other hand,
		\begin{align*}
			\frac{\partial V_u}{\partial x_3} & = -\sum_{\ell=1}^N\frac{(x_3-\frac{h}{2})}{|x-A_\ell|^3}  \\
			\frac{\partial V_d}{\partial x_3} & = -\sum_{\ell=1}^N\frac{(x_3+\frac{h}{2})}{|x-B_\ell|^3} ,
		\end{align*}
		and, again, along the line $x=2=0=x_3$ we have
		\begin{align*}
			\frac{\partial V_u}{\partial x_3} =  - \frac{\partial V_d}{\partial x_3} .
		\end{align*}
		Hence,
		\begin{align*}
			\langle \nabla V , n_1 \rangle  = 0 = \langle \nabla V , n_2 \rangle .
		\end{align*}
		Now, we shall finally consider the directional derivative in the direction of the tangent vector to the line under consideration $v=(1,0,0)$
		\begin{align*}
			\langle \nabla V , v \rangle & = \frac{\partial V}{\partial x_1} \\
			& = \frac{\partial V_u}{\partial x_1} + \frac{\partial V_d}{\partial x_1} \\
			& = - \sum_{\ell=1}^N \left( \frac{(x_1-R\cos(\ell \alpha + \frac{\alpha}{4}))}{|x-A_\ell|^3} + \frac{(x_1-R\cos(-\ell \alpha - \frac{\alpha}{4}))}{|x-B_\ell|^3} \right).
		\end{align*}
		Furthermore, along the line $x_2=0=x_3$, we have
		\begin{align*}
			\langle \nabla V , v \rangle & =  - \sum_{\ell=1}^N \left( \frac{(x_1-R\cos(\ell \alpha + \frac{\alpha}{4}))}{|x-A_\ell|^3} + \frac{(x_1-R\cos(-\ell \alpha - \frac{\alpha}{4}))}{|x-B_\ell|^3} \right) \\
			& =- 2\sum_{\ell=1}^N  \frac{(x_1-R\cos(\ell \alpha + \frac{\alpha}{4}))}{|x-A_\ell|^3} 
		\end{align*}
		and inserting $x_1=tR \cos(\frac{\alpha}{4})$ gives
		\begin{align*}
			\langle \nabla V , v \rangle & = - \frac{2}{R^2} \sum_{\ell=1}^N  \frac{ t \cos(\frac{\alpha}{4}) - \cos(\ell \alpha + \frac{\alpha}{4}) }{\left( (t \cos(\frac{\alpha}{4}) - \cos(\ell \alpha + \frac{\alpha}{4}))^2 + (\sin(\ell\alpha + \frac{\alpha}{4}) )^2 + \frac{h^2}{4}  \right)^\frac{3}{2}} .
		\end{align*}
		Now, notice that $\cos(\frac{\alpha}{4}) \geq \cos(\ell \alpha + \frac{\alpha}{4})$ with strict inequality for $\ell \neq N$, which shows that for $t=1$ we have $\langle \nabla V , v \rangle <0$. On the other hand, for evaluating this quantity near $t=0$, we compute its Taylor series at $t=0$
		\begin{align*}
			\langle \nabla V , v \rangle & =  \frac{16}{R^2} \frac{1}{\left( (4R^2 + h^2 \right)^\frac{3}{2}} \sum_{\ell=1}^N  \cos(\ell \alpha + \frac{\alpha}{4}) \\
			& \ \ \ \ + \frac{16 t}{R^2} \frac{\cos(\frac{\alpha}{4})}{(4R^2+h^2)^{\frac{5}{2}}} \sum_{\ell=1}^N \left( 12R^2 \cos(\ell\alpha + \frac{\alpha}{4})^2 - 4R^2 -h^2 \right) + O(t^2).
		\end{align*}
		Now, using the facts that
		$$\sum_{\ell=1}^N  \cos(\ell \alpha + \frac{\alpha}{4}) = 0 , \ \ \sum_{\ell=1}^N  \cos(\ell \alpha + \frac{\alpha}{4})^2 = \frac{N}{2},$$
		we find that 
		\begin{align*}
			\langle \nabla V , v \rangle & =  \frac{16 t}{R^2} \frac{\cos(\frac{\alpha}{4})}{(4R^2+h^2)^{\frac{5}{2}}}  \left( 6N R^2  - 4NR^2 -Nh^2 \right) + O(t^2) \\
			& = \frac{16N t}{R^2} \frac{\cos(\frac{\alpha}{4})}{(4R^2+h^2)^{\frac{5}{2}}}  \left( 2 R^2 - h^2 \right) + O(t^2) .
		\end{align*}
		Hence, if $h<\sqrt{2}h$ we have that along this line $\langle \nabla V , v \rangle$ is positive for small positive $t>0$. Combining this with the previously proven fact that it is negative for $t=1$ we find, from the intermediate value theorem, that $\langle \nabla V , v \rangle$ has a zero at some $t \in (0,1)$. This finishes the proof that $\nabla V$ as a zero along the segment connecting $0$ to $\frac{A_N+B_N}{2}$.
	\end{proof}
	
	For $h = \sqrt{2}R$, the complexity of the degenerate equilibrium at the origin is encapsulated by the Taylor expansion of the potential at the origin in homogeneous harmonic polynomials. 
	As we shall see, not only the degree-$2$ terms vanish,\footnote{We already knew this from the vanishing of the Hessian in the proof of Theorem \ref{thm:center_of_prism}.} but also the degree-$3$ terms do.
	\begin{lemma}
		\label{lem:Taylor_expansion_square}
		Let $V \colon \R^3 \to \R$ be the electrostatic potential generated by unit point charges placed at the vertices of the square anti-prism as specified at the beginning of this subsection.
		For $h=\sqrt{2}R$, its Taylor series at the origin is given by
		\begin{align}
			V(x_1, x_2, x_3) &= \tfrac {8 \sqrt {6}}{3} \left( 1  - \tfrac{7}{108} (x_1^4 + x_2^4) - \tfrac{14}{81} x_3^4 - \tfrac{7}{54} x_1^2x_2^2 + \tfrac{14}{27} x_3^2 (x_1^2 + x_2^2) \right) + \ldots ,
		\end{align}
		with terms of order exceeding $4$ not shown.
	\end{lemma}
	
	Again, we compute the Taylor expansion at the origin as a way to understand the degenerate critical point there for $h=\sqrt{2}R$.
	\begin{lemma}
		\label{lem:Taylor_expansion_triangle}
		Let $V \colon \R^3 \to \R$ be the electrostatic potential generated by unit point charges placed at the vertices of the triangular anti-prism as specified at the beginning of this subsection.
		For $h=\sqrt{2}R$, its Taylor series at the origin is given by
		\begin{align}
			\sqrt{3} - \tfrac{7\sqrt{3}}{216}\left( \tfrac{x_1^4}{8} + \tfrac{x_2^4}{8} + \tfrac{x_3^2}{3} + \tfrac{x_1^2x_2^2}{4} -x_1^2x_3^2 -x_2^2x_3^2 + \tfrac{5\sqrt{2}}{2} x_2x_3 ( x_2^2 - x_1^2 ) \right) + \ldots ,
		\end{align}
		with terms of order exceeding $4$ not shown.
	\end{lemma}

	\section{Proofs of Propositions~\ref{prop:no_maxima} and \ref{prop:octahedron}}
	\label{app:D}
	
	We first prove that, when $p \geq 1$, the function $V_p \colon \R^3 \setminus \{A_1,A_2,\ldots,A_n\} \to \R$ defined in \eqref{eqn:Vp} has no local maxima if all charges are positive.
	
	\begin{proof}[Proof of Proposition~\ref{prop:no_maxima}]
		The case $p=1$ follows immediately from the fact that $V_1=V$ is harmonic. 
		Hence, it suffices to consider the case $p>1$. 
		For this, we show that the Laplacian satisfies $\Delta V_p > 0$.
		Since it is the trace of the Hessian and thus the sum the diagonal eigenvalues, there must be at least one positive eigenvalue, which is incompatible with the existence of a maximum.
		To carry out the computation of $\Delta V_p$, we write $r_i(x) = \dist{x}{A_i}$ and thus $V_p(x) = \sum_{i=1}^n \zeta_i^p / r_i(x)^p$.
		Let $x_1, x_2, x_3$ be the coordinates in $\R^3$ and $A_i = (a_{i1}, a_{i2}, a_{i3})$.
		Then, using $\frac{\partial r_i}{\partial x_j} (x) = \frac{x_j-a_{ij}}{r_i}$, we compute 
		\begin{align}
			\frac{\partial V_p}{\partial x_j} 
			&= -p\sum\nolimits_{i=1}^n \frac{\zeta_i^{p}}{r_i^{p+1}} \frac{\partial r_i}{\partial x_j} 
			= -p\sum\nolimits_{i=1}^n \frac{\zeta_i^{p}}{r_i^{p+2}} (x_j-a_{ij}); \\
			\frac{\partial^2 V_p}{\partial x_j^2} 
			&= -p\sum\nolimits_{i=1}^n \zeta_i^{p}\left( \frac{1}{r_i^{p+2}} - (p+2) \frac{x_j-a_{ij}}{r_i^{p+3}} \frac{\partial r_i}{\partial x_j} \right) \\
			&= -p\sum\nolimits_{i=1}^n \zeta_i^{p}\left(  \frac{1}{r_i^{p+2}} - (p+2) \frac{(x_j-a_{ij})^2}{r_i^{p+4}}  \right) \\
			&= -p\sum\nolimits_{i=1}^n \frac{\zeta_i^{p}}{r_i^{p+2}}\left( 1 - (p+2) \frac{(x_j-a_{ij})^2}{r_i^{2}}  \right).
		\end{align}
		Finally, summing over the three coordinate directions, we find that
		\begin{align}
			\Delta V_p &= \sum\nolimits_{j=1}^3 \frac{\partial^2 V_p}{\partial x_j^2}
			= -p \sum\nolimits_{i=1}^n \frac{\zeta_i^{p}}{r_i^{p+2}} \left( 3 - (p+2) \sum\nolimits_{j=1}^3 \frac{(x_j-a_{ij})^2}{r_i^{2}}  \right) \\
			&= -p \sum\nolimits_{i=1}^n \frac{\zeta_i^{p}}{r_i^{p+2}} \left( 3 - (p+2) \right)
			= p (p-1) \sum\nolimits_{i=1}^n \frac{\zeta_i^{p}}{r_i^{p+2}} ,
		\end{align}
		which is positive if all $\zeta_i$ are positive.
	\end{proof}
	
	\begin{remark}
		\label{rem:no_maxima_no_minima}
		If one trades the assumption that $p \geq 1$ for $p \leq 1$ we can actually conclude, from the same proof, that there are no maxima.
		On the other hand, if one continues to assume that $p \geq 1$, but that instead that all $\zeta_i^p$ are negative, then the same proof shows that there are no minima when $p \geq 1$.
		These conclusions are particular cases of the following more general statement which can be inferred from the same proof. Suppose that $p>0$, then:
		\begin{itemize}
			\item if $(p-1)\zeta_i^p >0$ for all $i \in \{ 1, 2, \ldots , n \}$, then there are no local maxima;
			\item if $p=1$, then there are neither local minima nor minima;
			\item if $(p-1)\zeta_i^p <0$ for all $i \in \{ 1, 2, \ldots , n \}$, then there are no local minima.
		\end{itemize}
	\end{remark}
	
	Next, we shall prove that for $p>1$ the potentials $V_p$ can have minima when all the charges are positive. 
	This will be done by exploring the example given by placing the charges at the vertices of an octahedron. 
	\begin{proof}[Proof of Proposition~\ref{prop:octahedron}]
		With no loss of generality, we can choose the vertices of the octahedron at the points $(\pm 1, 0, 0)$, $(0,\pm 1,0)$, $(0,0,\pm 1)$, in which case the center of the octahedron coincides with the origin, $(0,0,0)$. 
		Then, using the computations carried out in the proof of Proposition~\ref{prop:no_maxima}, we find that
		\begin{align}
			\frac{\partial^2 V_p}{\partial x_i \partial x_j} (0,0,0) &= 0, \\
			\frac{\partial^2 V_p}{\partial x_i^2}(0,0,0) &= 2p (p-1),
		\end{align}
		for $i \neq j$ and $i \in \{ 1,2,3 \}$, respectively.
		Hence, the Hessian of $V_p$ at the origin is positive definite whenever $p>1$ and the result follows.
	\end{proof}


\end{document}